\documentclass[11pt]{article}
\usepackage{graphicx}    
\usepackage{tocloft}
\usepackage{waingarten}
\def\FullBox{\hbox{\vrule width 8pt height 8pt depth 0pt}}
\newcommand{\QED}{\;\;\;\FullBox}
\renewenvironment{proof}{\noindent{\bf Proof:~}}{\hfill\QED}
\newenvironment{proofof}[1]{\noindent{\bf Proof of {#1}:~}}{\hfill\(\QED\)}

\def\IS{\textsf{IS}_G}
\def\DO{\textsf{D}}
\def\ALG{\texttt{Alg}}

\title{Nearly optimal edge estimation with independent set queries}
\author{
Xi Chen\thanks{Columbia University, email: \texttt{xichen@cs.columbia.edu}.}
\and
Amit Levi\thanks{University of Waterloo, email: \texttt{amit.levi@uwaterloo.ca}.}
\and
Erik Waingarten\thanks{Columbia University, email: \texttt{eaw@cs.columbia.edu}.}}
\begin{document}         
\maketitle

\begin{abstract}
We study the problem of estimating the number of edges of an unknown, undirected graph $G=([n],E)$ with 
  access to an independent set oracle. 
When queried about a subset $S\subseteq [n]$ of vertices,
  the independent set oracle answers whether $S$ is an independent set in $G$ or not.
Our first main result is an algorithm that computes a $(1+\eps)$-approximation
  of the number of edges $m$ of the graph using $\min(\sqrt{m},n / \sqrt{m})\cdot\poly(\log n,1/\eps)$
  independent set queries.
This improves the upper bound of $\min(\sqrt{m},n^2/m)\cdot \poly(\log n,1/\eps)$
  by Beame et al. \cite{BHRRS18}.
Our second main result shows that ${\min(\sqrt{m},n/\sqrt{m}))/\polylog(n)}$ independent set queries are necessary, thus establishing that our algorithm is optimal up to a factor of ${\poly}(\log n, 1/\epsilon)$.

\end{abstract}

\thispagestyle{empty}

\newpage
\thispagestyle{empty}
\setcounter{tocdepth}{2}
\tableofcontents
\thispagestyle{empty}
\newpage
\setcounter{page}{1}

\newcommand{\ISO}{\texttt{IS}}
\newcommand{\BISO}{\texttt{BIS}}

\section{Introduction}

We study the problem of estimating the number of edges of a simple undirected graph $G = ([n], E)$ in the context of sublinear-time graph algorithms. 
The goal is to design a highly-efficient randomized algorithm that, 
given a certain type of oracle access to 
an underlying graph $G$, outputs a number $\wt{m} $ that approximates the number of edges of $G$. The first result in this direction was by Feige~\cite{F06},  
who studied this problem when the oracle is a \emph{degree~oracle}: the degree oracle
answers queries of the form ``what is the degree of a given vertex $v$?'' The algorithm of Feige makes $O(n/\sqrt{m})$ queries to the degree oracle, where $m$ denotes the number of edges of the input graph $G$,~and outputs a $(2+\eps)$-approximation to $m$ for any constant $\eps > 0$.
Moreover, Feige showed that the upper bound of $n/\sqrt{m} $ is tight for a $(2+\eps)$-approximation,
and indeed $\Omega(n^2/m)$ degree queries are necessary for a $(2-o(1))$-approximation.
Soon thereafter, Goldreich and Ron~\cite{GR08} considered~an oracle that, in addition to degree queries, can answer \emph{neighbor} queries (i.e., given a vertex $v\in [n]$ and an index $j$, the oracle returns the $j$th neighbor of $v$ according to some fixed ordering).
Their algorithm uses $\smash{\widetilde{O}(n/\sqrt{m})}$\hspace{0.04cm}\footnote{\label{footnoteone}We use $\widetilde{O}(f(n))$ and $\widetilde{\Omega}(f(n))$ to surpress $\polylog(f(n))$ factors.} queries and 
outputs a $(1+\eps)$-approximation to $m$ for any constant $\eps>0$;
they further showed that the upper bound is tight up to a $\polylog(n)$ factor.

Since then,  sublinear-time algorithms have been developed for a variety of 
graph problems, including estimating the number of stars \cite{GRS11,ABGPRY16}, triangles~\cite{ELRS17},
$k$-cliques \cite{ERS18}, and arbitrary small subgraphs \cite{AKK19},  finding forbidden graph minors \cite{KSS18,KSS19}, sampling edges almost uniformly \cite{ER18}, approximating the minimum weight spanning tree \cite{CRT05,CS09,CEFMNRS05}, maximum matching \cite{NO08,YYI09}, and minimum vertex cover \cite{PR07,MR09,NO08,YYI09,HKNO09,ORRR12}. 
As noted in a recent work of Beame, Har-Peled, Ramamoorthy, Rashtchian, and Sinha \cite{BHRRS18},
all these algorithms interact with oracles that provide only \emph{local} information about
the underlying graph (such as degree, neighbor, and \emph{edge existence} queries where an algorithm
can ask ``is vertex $u$ connected to vertex $v$?'')\footnote{One exception is
	that \cite{AKK19} also uses uniform edge sampling in addition to the above specified queries.}.
They suggested that 
\emph{non-local} oracle models may be natural in certain
scenarios of graph parameter estimation and their non-locality may enable
more efficient graph algorithms.

Along this line of investigation, \cite{BHRRS18} introduced both the \emph{independent set} oracle and the \emph{bipartite independent set} oracle {and studied the problem of estimating the number of edges under these two query models}. The independent set oracle for a graph $G = ([n], E)$ can be queried with a set $S \subseteq [n]$ of vertices  and outputs whether or not $S$ is an independent set in $G$, i.e. whether or not there exist vertices $u, v \in S$ with $(u, v) \in E$. The bipartite independent set oracle, on the other hand, can be queried with a pair of disjoint sets $S, T \subset [n]$ and outputs whether or~not~$(S, T)$ is a bipartite independent set in $G$, i.e. whether or not there exist $u \in S$ and $v \in T$ with $(u, v) \in E$.\footnote{We remark that the bipartite independent set oracle is at least as powerful, up to poly-logarithmic factors, as~the independent set oracle. Consider a graph $G = ([n], E)$, a set $S \subseteq [n]$ of vertices, and the question of whether or not $S$ is an independent set. Letting $(\bS_1,\bS_2)$ be a uniformly random partition of $S$, we may query the bipartite independent set oracle with $\bS_1$ and $\bS_2$. If $S$ is an independent set, then $(\bS_1, \bS_2)$ will be a bipartite independent set; if $S$ is not~an independent set, then $(\bS_1, \bS_2)$ will not be a bipartite independent set with probability at least $1/2$. Thus, $O(\log(1/\delta))$ bipartite independent set queries can simulate an independent set query with probability at least $1 - \delta$.
} 

{The problem of edge estimation using (bipartite) independent set queries 
	shares resemblance to the classical problem of \emph{group testing}, which dates back to 1943 \cite{D43}
	and has found many recent applications in computer science \cite{S85,CS90,DH00,ngo2000survey,MAP04,CM05,INR10}. 
	In group testing
	one needs to recover an unknown subset $S$ of a known universe $U$ by making \emph{subset queries}: an algorithm can 
	pick a subset $T$ of $U$ and ask  whether $T$ contains any element from $S$.
	The~graph setting of the current~paper~is~a natural generalization of group testing by considering the unknown object
	as a binary relation over a known universe $U$.
	The goal of estimating the number of edges, on the other hand, is a relaxation of group testing because it suffices to obtain
	an approximation of the size of the unknown binary relation, instead of recovering the relation itself exactly.
	The same relaxation on the original group testing setting (i.e., using subset queries to estimate the size of an unknown subset $S\subseteq U$)
	was studied by Ron and Tsur \cite{RT16}.
	Besides group testing, edge estimation using independent set queries 
	is motivated by connections to problems that arise in computational geometry and counting complexity,
	which we refer the interested reader to \cite{BHRRS18}.}

Perhaps surprisingly, \cite{BHRRS18} gave an algorithm that returns a $(1+\eps)$-approximation to the number of edges by making only $\poly(\log n, 1/\eps)$ queries to the \emph{bipartite} independent set oracle.~So in this setting, the non-locality indeed brings down the query complexity significantly for the edge estimation problem (compared to \cite{F06} and \cite{GR08}, both of which use local queries only).
For~the independent~set oracle,
\cite{BHRRS18} obtained an algorithm for a $(1+\eps)$-approximation 
of the number $m$ of edges  with query complexity $\min(\sqrt{m}, n^2/m) \cdot \poly(\log n, 1/\eps)$. %
It was left as an open problem in \cite{BHRRS18} to improve current understanding of 
edge estimation under independent set queries.

\subsection{Our results}


\begin{theorem}[Upper bound]\label{thm:intro1}
	There is a randomized algorithm that takes as input
	(1)~an accuracy parameter $\eps>0$, (2) a positive integer $n$ as the number of vertices and (3) access to
	the independent set oracle~of an undirected graph $G =([n],E)$ with $m=|E|\ge 1$.\footnote{The assumption of $m\ge 1$ is merely for convenience; it avoids the issue that the query complexity upper bound claimed would be $0$ when $m=0$.
		We note that whether a graph is empty or not can be determined by a single independent set query.}
	\hspace{-0.05cm}With probability at least $1-o(1)$, the algorithm makes
	no more than $\min(\sqrt{m},n/\sqrt{m})\cdot \poly(\log n, 1/\eps)$ many independent set queries and outputs  
	a number $\widetilde{m}$ that satisfies
	$
	(1-\eps) m\le \widetilde{m}\le (1+\eps) m$.
\end{theorem}

The improvement over the upper bound of \cite{BHRRS18} is due to a new algorithm for edge estimation that uses $(n/\sqrt{m})\cdot \poly(\log n,1/\eps)$ independent set queries
(Theorem \ref{thm:main}).
Note that the query complexity achieved by the algorithm underlying Theorem~\ref{thm:intro1} is essentially the same as \cite{GR08}; however, the two algorithms access the graph with very different ways 
(independent set oracle versus degree and neighbor oracles). The proof of Theorem~\ref{thm:intro1} requires new ideas and algorithmic techniques that are developed for independent set queries. 
See further discussion in Section \ref{sec:overview}.\medskip\vspace{0.05cm}


\begin{theorem}[Lower bound]\label{thm:intro2}
	Let $n$ and $m$ be two positive integers with $m\le {n\choose 2}$.
	Any randomized algorithm with access to the independent set oracle of an undirected graph 
	$G=([n],E)$ must make at least $\smash{{\min(\sqrt{m},n /\sqrt{m})/\polylog(n)}}$ queries in order to determine whether $|E|\le m/2$ or $|E|\ge m$~with probability at least $2/3$.
\end{theorem}

Theorems \ref{thm:intro1} and \ref{thm:intro2} essentially
settle the query complexity of edge estimation with independent set queries at $\min(\sqrt{m},n/\sqrt{m})$.
Theorem~\ref{thm:intro1} brings down the overall complexity~of the problem  from $n^{2/3}$ \cite{BHRRS18} to $\sqrt{n}$; the worst case is when the number of edges $m$ is linear in $n$. 
Theorem~\ref{thm:intro2},~on the other hand, shows that no algorithm with independent set queries
can achieve sub-polynomial query complexity.
This gives an exponential 
separation between the power of the bipartite independent set oracle and the independent set oracle
for the task of edge estimation.

\subsection{Overview of techniques}\label{sec:overview}

We first give a high-level overview of the lower bound because some key ideas from the lower bound will be helpful in understanding the main algorithm later. 
{For convenience we will~slightly~abuse~the notation $\widetilde{O}$ and $\widetilde{\Omega}$
	to hide
	factors of $\poly(\log n,1/\eps)$ in the discussion below.
	Outside of Section  \ref{sec:overview} they follow the convention 
	described in footnote \ref{footnoteone}.}
\subsubsection{Lower bound}
We describe our construction for the case when $m\ge n$, where we seek a lower bound of 
$\widetilde{\Omega}(n/\sqrt{m})$. The complement case follows from a  reduction to this case. 

The plan is to follow Yao's principle. We construct two distributions $\Dyes$ and $\Dno$ over
graphs with vertices $[n]$ so that $\bG\sim \Dyes$ has no more than $m/2$ edges with
probability at least $1-o(1)$ and $\bG\sim \Dno$ has at least $m$ edges with probability at least $1-o(1)$.
We then show that no deterministic algorithm with access to an independent set oracle
can distinguish these two distributions.

A graph $\bG\sim\Dyes$ is generated by first sampling a uniformly random partition of vertices into $(\bA, \ol{\bA})$ and then forming the bipartite graph by including each pair $(i, j)$ with $i \in \bA$ and $j \in \ol{\bA}$ as an edge independently with probability $d/n$, where $\smash{d\eqdef m/n}$. 
In expectation 
$\bG \sim \Dyes$ has about $m/4$ edges and thus, has no more than $m/2$ edges with probability $1-o(1)$.
On the other hand, a graph $\bG\sim\Dno$ is generated by sampling a uniformly random partition $(\bA, \ol{\bA})$ of $[n]$, as well as a subset $\bB \subseteq \bA$ by including each vertex of $\bA$ independently with probability $d \log n / n$. Similar  to $\Dyes$, a pair $(i, j)$ where $i \in \bA \setminus \bB$ and $j \in \ol{\bA}$  is included as an edge independently with probability $d/n$. 
The main difference compared to $\Dyes$ is that every pair $(i, j)$, where $i \in \bB$ and $j \in \ol{\bA}$, is included as an edge (so $(\bB, \ol{\bA})$ form a complete bipartite graph). 
Given that $|\ol{\bA}|=\Omega(n)$ and $|\bB|=\Omega(d\log n)$ with high probability,
the number of edges in the graph is $\Omega(dn\log n)=\Omega(m\log n)\ge m$ with  probability
at least $1-o(1)$.

We make the following two observations.
The first is that a graph $\bG\sim \Dno$ can be generated~by~first
drawing a graph $\bG'\sim \Dyes$ with partition $(\bA,\ol{\bA})$, then sampling $\bB\subseteq \bA$
by including each vertex in $\bA$ independently with probability $d\log n/n$, and finally adding
all pairs between $\bB$ and $\ol{\bA}$ as edges~in $\bG$.
This suggests that, in order for an algorithm to distinguish $\Dno$ from $\Dyes$,
a (seemingly quite weak) necessary condition is for one of its queries to overlap with $\bB$
when it runs on $\bG\sim \Dno$.

For the second observation,  
we consider a query set $S\subseteq [n]$ of size larger than $(n/\sqrt{m})\cdot \log n$.
In both $\Dyes$ and $\Dno$, we have $|S\cap\bA|,\hspace{0.03cm}|S\cap \ol{\bA}|\ge\Omega((n/\sqrt{m})\cdot \log n)$ with high probability and~when this happens,
$S$ is not an independent set with high probability, given that there are~at least
$$\Omega\left((n^2/m)\cdot\log^2 n\right)=\Omega\left((n/d)\cdot \log^2 n\right)$$ pairs between
$S\cap \bA$ and $S\cap \ol{\bA}$ and each is included in the graph with probability $d/n$.
Since $S$ is not an independent set in both $\Dyes$ and $\Dno$ with high probability,
such a query conveys very little information in distinguishing the two distributions.
Thus, a reasonable algorithm should only make queries
of size smaller than $(n/\sqrt{m})\cdot \log n$. This intuition, that algorithms should not make queries of~size larger than $n/\sqrt{m}$, will be helpful in our discussion of the algorithm later, and we will frequently refer to the quantity $n / \sqrt{m}$ as the \emph{critical threshold}.
However, if all the queries an algorithm makes are smaller than $(n/\sqrt{m}) \cdot \log n$, then 
$\smash{\widetilde{\Omega}(n/\sqrt{m})}$ queries are necessary for at least one of them to overlap with $\bB$;
otherwise, given that $|\bB|=O(d\log n)$,
the probability that one of the queries overlaps with $\bB$ is negligible.

To formalize the above intuition and simplify the presentation of our lower bound proof,
we introduce the notion of an \emph{augmented} (independent set) oracle in Section \ref{sec:proof}. 
We first show that any algorithm with access to the standard independent set oracle
can be simulated using an augmented oracle with the same query complexity. Then, we prove an $\smash{\widetilde{\Omega}(n/\sqrt{m})}$ lower bound 
for algorithms that distinguish $\Dyes$ and $\Dno$ with access to an augmented oracle.
\subsubsection{Upper Bound}
Our goal is to obtain a $(1+\eps)$-approximation algorithm for edge estimation 
with $ \widetilde{O}(n/\sqrt{m})$ independent set queries, where $m$ denotes
the number of edges~of~the~input graph (Theorem \ref{thm:main}). 
Theorem \ref{thm:intro1} follows by combining it with the 
algorithm of \cite{BHRRS18} by running both algorithms in parallel and 
outputting the result of whichever finishes first.

In the sketch of the algorithm below, we assume that a rough estimate $\ol{m}$ of the number 
of edges $m$ is given, satisfying $m=\Theta(\ol{m})$. The goal is to refine it to
obtain a $(1+\eps)$-approximation $\widetilde{m}$ of $m$.


\paragraph{An Initial Plan:}

At a high level, we partition the vertex set $[n]$ into $O((\log n)/\eps)$ many buckets
according to their degrees: a vertex $u\in [n]$ belongs to the $i$th bucket $B_i$ if $\deg(u)$ is between $(1+\eps)^i$ and $(1+\eps)^{i+1}$. We refer to $(1+\eps)^i$ as the degree of bucket $B_i$ for convenience.
Our initial plan is to develop efficient algorithms for the following two tasks:
\begin{enumerate}
	\item[]\hspace{-0.7cm} \textbf{Task 1}: Develop a subroutine that, given a vertex $u$
	and an index $i$, 
	checks if $u$ belongs to $B_i$.\footnote{The goal of the subroutine as described
		above may not sound reasonable. If $\deg(u)$ lies very close to the boundary of two buckets $B_i$ and $B_{i+1}$, determining which of the two buckets $u$ lies in may be expensive with independent set queries. 
		This is indeed one source of errors we need to handle. 
		We focus on high-level ideas behind the algorithm and skip details such as errors
		most of time, and discuss briefly how we analyze the algorithm in the presence of errors
		at the end of the sketch.}
	\vspace{-0.25cm}
	
	\item[]\hspace{-0.7cm} \textbf{Task 2}: Use  the first subroutine 
	to estimate the size of each bucket $B_i $. 
\end{enumerate}
We point out that this initial plan looks very similar to the framework of the algorithm of \cite{GR08},
where ideally one would like to estimate the size of each $B_i$ by drawing
enough random samples and running the subroutine in Task 1 on each sample to obtain an 
estimate of $|B_i|$.
The similarity, however, stops here as we start discussing more details about how to implement the plan
with an independent set oracle.

We consider Task 1 first (which is trivial with a degree oracle). 
Note that when $d\ge \sqrt{\ol{m}}$, 
checking whether a vertex $u$ has $\deg(u)\ge (1+\eps)d$ or $\deg(u)\le d$ 
requires 
$\widetilde{O}(1)$ independent set queries.
As a result, it requires $\widetilde{O}(1)$ to tell if $u\in B_i$ 
when the degree of $B_i$ is at least $\sqrt{\ol{m}}$.
The bad news is that the same task  becomes significantly more challenging as $d$ goes down from $\sqrt{\ol{m}}$.
This challenge leads to a major revision of our initial plan.   

To gain some intuition we consider the task of distinguishing  
$\deg(u)\ge (1+\eps)d$ and $\deg(u)\le d$ when $d\gg \sqrt{m}$.\footnote{For convenience we consider the case of
	$d\gg\sqrt{m}$ in the sketch but the same idea works when $d\ge \sqrt{m}$.}
Suppose we sample a set $\bT$ from $[n]\setminus \{u\}$ by including each vertex with probability $1/d$ and then make two independent set queries on $\bT$ and $\bT\cup \{u\}$. 
Let $\cal{E}$ denote the event that~$\bT$ is an independent set but $\bT\cup \{u\}$ is not
(so $\bT$ contains at least one neighbor of $u$).
Then we claim that there is a significant gap 
in the probability of $\cal{E}$ when $\deg(u)\ge (1+\eps)d$ versus $\deg(u)\le d$.
This gap in the probability of $\calE$ is large enough so that one can repeat the experiment $\widetilde{O}(1)$ times (each time making two independent set queries)
to distinguish the two cases with high probability.


Now we turn to the case when 
$d\ll \sqrt{\ol{m}}$.
In this case, the algorithm is limited to query sets $\bT$ of size much smaller than $n/d$.
Therefore,
we limit $\bT$ to include each vertex with probability $1/\sqrt{\ol{m}}$ instead of $1/d$.
Two issues arise. 
The first (minor) issue is that, given that the size of $\bT$ is roughly $n/\sqrt{\ol{m}}$,
even to hit a neighbor of $u$ (with degree roughly $d$) one needs to draw 
$\bT$ at least $ \sqrt{\ol{m}}/d$ many times.
This suggests that $ \sqrt{\ol{m}}/d$ queries are needed for Task~1 when the degree
$d$ of the bucket we are interested is less than $\sqrt{\ol{m}}$.  


There is, however, a more serious issue that is subtle but leads to a major revision of the initial plan.
Consider the scenario where $u$ has $(1+\eps)d$ neighbors and every neighbor has degree $\gg \sqrt{\ol{m}}$.
If we sample $\bT$ by including each vertex with probability $1/\sqrt{\ol{m}}$,
it is very unlikely that $\bT$ contains a neighbor of $u$ but $\bT$ is at the same time independent
(since when conditioning on $\bT$ containing a neighbor $v$ of $u$,  most likely $\bT$ also
contains a neighbor of $v$ given the large degree of $v$).
Because~of the second issue, we change the goal of the subroutine in Task 1 from
finding the right bucket of $u$ according to the degree of $u$ to finding
the right bucket according to the \emph{number of neighbors of $u$ with degree at most} $\sqrt{\ol{m}}$,
when $\deg(u)<\sqrt{\ol{m}}$.
For vertices with degree at least $\sqrt{\ol{m}}$, we still would like to partition them into buckets
according to their degrees. 

\paragraph{A Revised Plan:} By the above, we arrived at the following revised plan:
\begin{flushleft}\begin{enumerate}
		\item[]\hspace{-0.7cm} \textbf{Task 0}: Develop a subroutine that, given a vertex $u$, decides\footnote{Again we need to handle errors when $\deg(u)$ is close to $\sqrt{\ol{m}}$.} if $\deg(u)\ge \sqrt{\ol{m}}$ (which we refer \\to as
		high-degree vertices and denote the set by $H$) or $\deg(u)<\sqrt{m}$ (which we refer to as low-degree vertices
		and denote the set by $L$).
		High-degree vertices are further partitioned into buckets $H_i$ according to their degrees.
		Low-degree vertices, on the other hand, are partitioned into buckets $L_i$ according to their 
		degrees to low-degree vertices, denoted by $\deg(u,L)$ for a vertex $u$.\vspace{-0.2cm}
		\item[]\hspace{-0.7cm} \textbf{Task 1}: 
		Develop a subroutine that, given a vertex $u\in H$ (or $u\in L$) and an index $i$,
		decides\\ if $u$ belongs to the bucket $H_i$ (or $L_i$).
		\vspace{-0.2cm}

		\item[]\hspace{-0.7cm} \textbf{Task 2}: Use  the two subroutines 
		to obtain $(1+\eps)$-estimations of the size of each $L_i$ and $H_i$. 
\end{enumerate}\end{flushleft}

Looking ahead, with $(1+\eps)$-approximations $\ell_i$ and $h_i$ for $|L_i|$ and $|H_i|$,
one can compute \begin{eqnarray*}&\sum_i \ell_i\cdot (1+\eps)^i +\sum_i h_i\cdot (1+\eps)^i&\end{eqnarray*}
as roughly a $2$-approximation of the number of edges $m$.
The reason that we only get $2$-approximation follows by the fact that in the sum,
edges between  vertices in $L$ and edges between vertices in $H$ are counted twice 
but edges between $L$ and $H$ are only counted once.
We will discuss more about how to further revise the plan to obtain a $(1+\eps)$-approximation; for now let us consider Task 2.

Note that Task 2 for buckets $L_i$ is easy.
Consider a low-degree bucket $L_i$ with $d=(1+\eps)^i\le \sqrt{\ol{m}}$.
Unless $|L_i|=\Omega(\ol{m}/d)$, $L_i$ has negligible impact on the final estimate.
When $|L_i|=\Omega(\ol{m}/d)$, it takes $\widetilde{O}(nd/\ol{m})$ samples to 
get a sufficient number~of vertices in $L_i$.
We can then get a good estimation of $|L_i|$ by 
running subroutines for Task 0 and 1 on these vertices.
We pay $\widetilde{O}(\sqrt{\ol{m}}/d)$ queries for each vertex so the overall query complexity
is 
$$
\widetilde{O}(nd/\ol{m})\cdot \widetilde{O}(\sqrt{\ol{m}}/d)=\widetilde{O}(n/\sqrt{\ol{m}})
$$
as desired.
In contrast,  uniformly sampling vertices and checking individually if each of them
lies~in $H_i$ is too inefficient for high-degree buckets,
given that $nd/\ol{m}\gg n/\sqrt{\ol{m}}$ when $d\gg \sqrt{\ol{m}}$. 

Estimating the size of each high-degree bucket $H_i$ is where we fully take advantage of
the \emph{non-locality} of independent set queries.
To explain the intuition, let us consider the task of distinguishing $|H_i|\ge (1+\eps)r$ versus $|H_i|\le r$ for 
some parameter $r=\Theta(\ol{m}/d)$ where $d=(1+\eps)^i\gg \sqrt{\ol{m}}$ denotes the degree of the bucket $H_i$.
To this end, it suffices to~have a procedure that can take
a random set $\bS\subseteq [n]$ of size $n/(\sqrt{\ol{m}} \log n)$ and 
answers the question ``does there exist $u\in \bS$ that belongs to $H_i$?'' with $\widetilde{O}(1)$
queries.
With such a procedure it suffices to draw $\bS$ and run the procedure on $\bS$ for 
$$
\widetilde{O}\left(\frac{n}{( {n}/({\sqrt{\ol{m} }\log n))}\cdot {(\ol{m}}/{d)}}\right)=\widetilde{O}
\left(\frac{d}{\sqrt{\ol{m}}}\right)\le \widetilde{O}\left(\frac{n}{\sqrt{\ol{m}}}\right)
$$
many times in order to obtain a good estimation of $|H_i|$.

As discussed earlier, 
the revised plan ultimately leads to a $(2+\eps)$-approximation
algorithm with $\smash{\widetilde{O}(n/\sqrt{m})}$ independent set queries.
We achieve $(1+\eps)$-approximation by revising the plan~further.
First we  divide high-degree vertices $u$ into buckets $H_{i,j}$ where $i$ 
is related to the degree of $u$ (as usual), but the second index $j$ is related 
to the fraction of neighbors of $u$ in $L$; see Definition \ref{def:degree-part} for details. 
Task 1 is updated to develop a subroutine that can decide whether $u$ belongs to $H_{i,j}$~or not.  
Task 2 is updated to estimate the size of each $H_{i,j}$
(with similar ideas in the approximation of $|H_i|$ sketched above) and $L_i$.
Together they lead to a $(1+\eps)$-approximation of the number of edges between low-degree and high-degree vertices,
and ultimately a $(1+\eps)$-approximation of $m$.

Now extra care must be taken to handle errors 
when executing the above plan.
As~alerted~in~two footnotes, one cannot hope for a subroutine
that returns the true bucket of a vertex $u$.
To simplify the presentation of the algorithm and its analysis, we 
introduce the notion of $(\ol{m},\eps)$-\emph{degree~oracles}
(see Definition \ref{def:degree-oracle}).
An $(\ol{m},\eps)$-degree oracle can answer questions listed in Tasks 0 and 1 
consistently and accurately up to certain errors (as captured by the notion of 
an $(\ol{m},\eps)$-\emph{degree partition} in Definition \ref{def:degree-part} underlying each  $(\ol{m},\eps)$-{degree oracle}).
We first present an algorithm  in~Section~\ref{sec:estimate}~that 
has query access to a $(\ol{m},\eps)$-{degree oracle}.
We finish the proof of Theorem \ref{thm:main} by giving an efficient implementation of a $(\ol{m},\eps)$-{degree oracle} using an independent set oracle in Section \ref{sec:check-degrees}.

\section{Preliminaries}

Given a positive integer $n$, we write $[n]$ to denote $\{1,\ldots,n\}$. Similarly, for two non-negative integers $i\le j$, we write $[i:j]$ to denote $\{i,\ldots,j\}$.
All graphs considered in this paper are undirected and simple 
(meaning that there are no parallel edges or loops),
and have $[n]$ as its vertex set. 

\begin{definition}[Independent set oracle]
	Given an undirected graph $G = ([n], E)$, its \emph{independent set oracle} is a map $\IS \colon 2^{[n]} \to \{0,1\}$ which satisfies that for any set of vertices $U \subseteq [n]$, $\IS(U) = 1$ if~and only if $U$ is an independent set of $G$ \emph{(}i.e., $(u, v) \notin E$ for all $u, v \in U$\emph{)}.
\end{definition}


\newcommand{\BinarySearch}{\texttt{Binary-Search}}
We use $\deg_G(v)$ to denote the degree of a vertex $v\in [n]$.
Given $v\in [n]$ and $U\subseteq [n]$, we let
$$\Gamma_G(v,U)=\big\{u\in U: (u,v)\in E\big\}\quad\text{and}\quad \deg_G(v,U)\eqdef\big|\Gamma_G(v,U)\big|.$$  
Note that $v$ can lie in $U$, but since we only consider simple graphs, $\Gamma_G(v,U)=\Gamma_G(v,U\setminus \{v\})$.
For the sake of brevity, we write $\Gamma_G(v)=\Gamma_G(v,[n])$. We usually skip the subscript in $\IS,\Gamma_G$ and $\deg_G$ when the underlying graph $G$ is clear from the context.

The following simple lemma will be used multiple times.
\begin{lemma}\label{lem:indset}
	Let $G=([n],E)$ be an undirected graph, $S\subseteq [n]$ be a 
	set of vertices, and $r \in \N$ be {an upper bound on the number of edges} 
	in the subgraph induced by $S$.
	Let $\bT\subseteq S$ be a random subset given by independently including each vertex of $S$ with probability $p$. Then,
	\[ \Prx_{\bT \subseteq S}\big[ \bT \text{ is an independent set of $G$ } \big] \geq 1 - rp^2. \]
\end{lemma}
\begin{proof}
	The expected number of edges where both vertices lie in  $\bT$ is at most $rp^2$. By Markov's inequality the probability that $\bT$ contains at least one edge is at most $rp^2$.
\end{proof}


\subsection{Binary search using the independent set oracle}\label{sec:binary}

\begin{figure}[t!]
	\begin{framed}
		\noindent Subroutine $\BinarySearch\hspace{0.04cm}(n,G,T,\delta)$
		\begin{flushleft}
			\noindent {\bf Input:} A positive integer $n$, access to 
			the independent set oracle of a graph $G = ([n], E)$, a set $T\subseteq [n]$ with a promise that $T$ is not an independent set of $G$, {and an error parameter $\delta > 0$}.\\
			{\bf Output:} An edge $(u, v) \in E$ with $u,v\in T$, or ``fail.'' \\
			\begin{enumerate}
				%
				
				\item {Let $\bA \leftarrow T$}.
				
				\item Repeat the following for $ t=O(\log n+\log (1/\delta))$ iterations:
				\begin{enumerate}
					\item If $|\bA|=2$, output the two vertices in $\bA$ 
					\item Randomly partition $\bA$ into $\bA_1 \cup \bA_2$ where $|\bA_1|$ and $|\bA_2|$ differ by at most 1. \\Query $\IS(\bA_1)$ and $\IS(\bA_2)$ to see if one of them is not an independent set.
					\\If $\bA_b$ is not an independent set for some $b\in \{1,2\}$, 
					set $\bA \leftarrow \bA_b$.
				\end{enumerate}
				\item Output ``fail''. 
			\end{enumerate}

		\end{flushleft}\vskip -0.14in
	\end{framed}\vspace{-0.25cm}
	\caption{Description of the $\BinarySearch$ subroutine.} \label{fig:binary-search}
\end{figure}

We present a subroutine based on binary search for finding an edge
using independent set queries:

\begin{lemma}\label{lem:binary-search}
	There is a randomized algorithm, $\emph{\BinarySearch}\hspace{0.04cm}(n,G,T,\delta),$ that takes as input
	(1) a posi\-tive integer $n$,
	(2) access to the independent set oracle $\IS$ of an undirected graph $G=([n],E)$,
	(3) a~set $T\subseteq [n]$ of vertices such that $T$ is not an independent set of $G$,
	and (4) an error parameter $\delta>0$.
	$\emph{\BinarySearch}$ makes 
	$O(\log n+\log (1/\delta))$ queries to $\IS$ and outputs $u,v \in T$ with $(u, v) \in E$ 
	with probability at least $1 - \delta$.
\end{lemma}
\begin{proof} We consider an execution of $\BinarySearch(n, G, T, \delta)$ in Figure~\ref{fig:binary-search}. Note that we maintain the invariant that $\bA$ is never an independent set. This is because $T$ is not an independent set in step 1, and whenever $\bA$ is updated in step 2(b), it is never assigned an independent set. It suffices to show that after $t$ iterations, $|\bA| = 2$ with high probability.
	
	An iteration of step 2 makes progress if the size of the set $\bA$ decreases by {at least constant factor}.
	If, in any iteration of step 2(b), the partition of $\bA$ into $\bA_1$ and $\bA_2$ has at least one edge fully contained in $\bA_1$ or $\bA_2$, then that iteration will make progress. Since there is always at least one edge in $\bA$, this occurs independently in each iteration with probability at least $1/2$. Since it only takes $O(\log n)$ rounds for the size of $\bA$ to drop to $2$, it follows from Chernoff bound that the subroutine fails with probability at most $\delta$.
\end{proof}

\begin{remark} We will {always} invoke \emph{\BinarySearch} with the parameter $\delta=1/\poly(n)$.\footnote{{For example, setting $\delta = 1/n^{10}$ will suffice for our purposes.}} The subroutine will always make $O(\log n)$ 
	queries, and~will fail with probability at most $1/\poly(n)$.\end{remark}

\section{Upper bound}\label{sec:upperbound}
\newcommand{\mbar}{\overline{m}}
\newcommand{\CheckDegree}{\texttt{Check-Degree}}
\newcommand{\CheckHiDegree}{\texttt{Check-High-Degree}}
\newcommand{\CheckLoDegree}{\texttt{Check-Low-Degree}}
\newcommand{\CheckH}{\texttt{CheckH}}
\newcommand{\CheckL}{\texttt{CheckL}}
\newcommand{\HiDegEvent}{\texttt{High-Degree-Event}}
\newcommand{\HiDegBucket}{\texttt{High-Degree-Bucket}}
\newcommand{\LowDegEvent}{\texttt{Low-Degree-Event}}
\newcommand{\LowDegBucket}{\texttt{Low-Degree-Bucket}}
\newcommand{\TwoApprox}{\texttt{Estimate-With-Advice}}
\newcommand{\TwoApproxNoAdvice}{\texttt{Estimate-Edges}}


In this section we prove the following upper bound:

\begin{theorem} \label{thm:main}
	There is a randomized algorithm $\emph{\TwoApproxNoAdvice}\hspace{0.04cm}(\eps,n,G)$ that takes as input
	(1)~an accuracy parameter $\eps\in (0,1)$, (2) a positive integer $n$, and (3) access to
	the independent set oracle of a  graph $G=([n],E)$ with $m=|E|\ge 1$. 
	With probability at least $1-o(1)$, $\emph{\TwoApproxNoAdvice}$~makes
	$(n/\sqrt{m})\cdot \poly(\log n, 1/\eps)$ queries and outputs  
	a number $\widetilde{m}$ satisfying
	$(1-\eps) m\le \widetilde{m}\le (1+\eps) m.$
\end{theorem}

{We recall the following lemma from \cite{BHRRS18}.
	\begin{lemma}[Lemma~5.6 from \cite{BHRRS18}]\label{lem:bhhrs}
		There is a randomized algorithm that takes as input (1) an accuracy parameter $\eps \in (0,1)$, (2) a positive integer $n$, and (3) access to the independent set oracle of a graph $G = ([n], E)$ with $m = |E| \geq 1$. With probability at least $1-o(1)$, the algorithm makes $\sqrt{m} \cdot \poly(\log n, 1/\eps)$ queries and outputs a number $\wt{m}$ satisfying $(1-\eps) m \leq \wt{m} \leq (1+\eps) m$.
	\end{lemma}
	
	The upper bound claimed in Theorem~\ref{thm:intro1} of $\min\{ \sqrt{m}, n / \sqrt{m}\} \cdot \poly(\log n, 1/\eps)$ follows by running the algorithm of Theorem~\ref{thm:main} and the algorithm of Lemma~\ref{lem:bhhrs} in parallel. Specifically, we alternate queries between the two algorithms until one of them terminates. Once one terminates~with~an~estimate $\wt{m}$ to $m$, we output $\wt{m}$. }

\subsection{Reduction to edge estimation with advice}


We prove Theorem \ref{thm:main} using the following lemma stated next. We will provide an algorithm, which we call
{\tt Estimate-With\--Advice}, for estimating $|E|$ given an extra parameter $\ol{m}$ which is promised to be an upper bound for $|E|$.

\begin{lemma}[Estimation with advice]\label{lem:up-given}
	There is a randomized algorithm, $\emph{\TwoApprox}$, 
	that takes four inputs: (1) an accuracy 
	parameter $\eps\in (0,1)$,
	(2) two positive integers $n,\ol{m}$, and~(3) \mbox{access~to~an} 
	independent set oracle of $G=([n],E)$ with $1 \leq m=|E|\le \ol{m}$.
	$\emph{\TwoApprox} $ makes $(n/\sqrt{\ol{m}})\cdot \poly(\log n,1/\eps)$ 
	queries and with probability at least $1-1/n$ outputs $\hat{m}$ that satisfies
	\begin{equation}\label{output}
	(1-5\eps) {m} -O\left(\frac{\eps \ol{m}}{\log n}\right)\le \hat{m}\le (1+\eps) m.
	\end{equation}
\end{lemma}

Before proving Lemma \ref{lem:up-given}, we show that it implies Theorem~\ref{thm:main}.\medskip\vspace{0.06cm}

\begin{figure}[t!]
	\begin{framed}
		\noindent Algorithm $\TwoApproxNoAdvice\hspace{0.05cm}(\eps,n,G)$
		\begin{flushleft}
			\noindent {\bf Input:} An accuracy parameter $\eps\in (0,1)$, a positive integer
			$n$, and access to 
			the independent set oracle of an undirected graph $G=([n],E)$.\\ 
			{\bf Output:} A number $\widetilde{m}$ as an estimation of $m=|E|$.
			\begin{enumerate}
				\item Set $\ol{m}={n\choose 2}$.\vspace{-0.06cm}
				\item While $\ol{m}\ge  1$: \vspace{-0.05cm}				\begin{enumerate}
					\item Invoke $\TwoApprox\hspace{0.04cm}(\epsilon/11 ,n,\ol{m},G)$.\vspace{0.03cm}
					\item Let $\hat{m}$ denote the output. If $4\hspace{0.02cm}\hat{m}\ge \mbar $ {\bf return} $\hat{m}$ as $\widetilde{m}$;
					otherwise set $\ol{m}$ to be $\lfloor \ol{m}/2\rfloor$. 
				\end{enumerate}
				\item {\bf Return} $0$ as $\widetilde{m}$ (this line is reached with low probability). 
			\end{enumerate}
			
		\end{flushleft}\vskip -0.14in
	\end{framed}\vspace{-0.25cm}
	\caption{Description of the $\TwoApproxNoAdvice$ algorithm.\vspace{-0.25cm}} \label{fig:Two-Approx-No-Advice}
\end{figure}
\begin{proofof} {Theorem~\ref{thm:main} Assuming Lemma~\ref{lem:up-given}} 
	We present $\TwoApproxNoAdvice$ in Figure \ref{fig:Two-Approx-No-Advice}.

	Note that at the end of each iteration of step 2 in Figure~\ref{fig:Two-Approx-No-Advice},
	either the algorithm terminates or $\ol{m}$ is halved. Since $\ol{m}$ is initially $\binom{n}{2}$, the maximum number of iterations of the step 2 (before $\ol{m} < 1$) in $\TwoApproxNoAdvice$ is $O(\log n)$.
	It follows from Lemma~\ref{lem:up-given} and a union bound  that, with probability at least 
	$1-o(1)$, every execution of 
	$\TwoApprox$ in step 2(a) of $\TwoApproxNoAdvice$ returns a \emph{correct} value
	(meaning that if $\ol{m}$ of this run indeed satisfies $\ol{m}\ge m=|E|$, then its output $\hat{m}$
	satisfies (\ref{output}) but with $\eps$ set to $\eps/11$).
	We show  that the following holds when this is the case:
	\begin{flushleft}\begin{enumerate}
			\item[] ($*$): $\TwoApproxNoAdvice$ terminates in the while loop (instead of going to line 3)\\ with the 
			final value of $\ol{m}$ satisfying $m\le \ol{m}\le 5m$.
	\end{enumerate}\end{flushleft}
	{Assume that ($*$) holds, and let $\wt{m}$ be the output of $\TwoApproxNoAdvice(\eps, n, G)$. Since $m \leq \ol{m}$ in every iteration of step 2 and the final iteration also satisfies $\ol{m}\leq 5m$, Theorem \ref{thm:main} would follow from two observations.}
	(i) The query complexity of $\TwoApproxNoAdvice$ can be bounded using $\ol{m}\ge m$, and (ii) since the final run of $\TwoApprox$ is correct, we have (using $\ol{m}\le 5m$)
	\begin{equation}\label{samesame}
	(1-\eps/2)m< (1-5\eps/11) m -O\left(\frac{\eps m}{\log n}\right)\le 
	\widetilde{m}\le (1+\eps/{11} )m< (1+\eps)m.
	\end{equation}
	It suffices to show that 
	($*$) holds when every run of $\TwoApprox$ returns a correct value.
	
	Assuming for contradiction of ($*$) that the final value of $\ol{m}$ is smaller than $m$.
	This implies that $m\le \ol{m}$ $\le 2m$ in one of the runs of $\TwoApprox$ in $\TwoApproxNoAdvice$.
	Since it returns a correct value $\hat{m}$ (and note that for this run we still have $m\le \ol{m}$), 
	the same calculation in (\ref{samesame}) implies that $\hat{m}\ge (1-\eps/2)m$ and thus,
	$4\hat{m}\ge 4\cdot 0.5\cdot m=2m\ge \ol{m}$ and the algorithm should have terminated at the end of this run,
	a contradiction.
	On the other hand, assume for a contradiction of ($*$) that the final value of $\ol{m}$ is larger than $5m$.
	Since the final run returns a correct value $\hat{m}$, 
	$\hat{m}\le (1+\eps/{11})m\le 1.1m$
	and thus, $4\hat{m}<5m<\ol{m}$; {however, step 2(b) should have terminated if $4\hat{m} \geq \ol{m}$,} a contradiction. This finishes the proof of the theorem.  
\end{proofof}\medskip\vspace{0.1cm}

We prove Lemma \ref{lem:up-given} in the rest of the section.
From now on, let $\eps\in (0,1)$ be the accuracy~parameter,  
$\ol{m}\le {n\choose 2}$ be a positive integer, and $G=([n],E)$ be a 
graph with $1\le m=|E|\le \ol{m}$ as in the statement of Lemma \ref{lem:up-given}.
Let $\alpha = 1+\eps$ and let $s$ be the unique positive integer such that 
\begin{align}
	\alpha^{s-1}&\le \sqrt{\ol{m}}<\alpha^s. \label{eq:def-s}
\end{align}
We also write $\beta=\Theta((\log n)/\eps)$ to denote the smallest integer such that 
$\alpha^\beta\ge n$, and $\tau$ to denote the smallest integer 
such that {$\alpha^\tau\ge \log n/\eps$ (so $\alpha^\tau = \Theta(\log n/\eps)$)}.
It may be helpful to the reader to consider the case when $\ol{m}$ is only a constant factor larger than $|E|$, so the algorithm's task is to refine an approximation to the number of edges given a crude approximation; however, the proof of Lemma~\ref{lem:up-given} assumes just the upper bound $\ol{m} \geq |E|$.

\subsection{Degree oracles and the high-level plan}\label{sec:oracles-and-plan}

To simplify the presentation and analysis of our algorithm, $\TwoApprox$, we introduce
the notion of \emph{$(\ol{m},\eps)$-degree partitions} and \emph{$(\ol{m},\eps)$-degree oracles}.
Roughly speaking, an $(\ol{m},\eps)$-degree partition 
$P=(L_i,H_{k,\ell}: i\in [0:s], k\in [s+1:\beta]\ \text{and}\ \ell\in [0:\tau])$ of an undirected graph $G=([n],E)$
is a partition of  $[n]$ (so $L_i$'s and $H_{k,\ell}$'s  
are pairwise disjoint subsets of $[n]$ whose union is $[n]$)
such that the placement of a vertex $v$ reveals important degree information of $v$
(see Definition \ref{def:degree-part} for details).
An $(\ol{m},\eps)$-degree oracle, on the other hand, contains an underlying
$(\ol{m},\eps)$-degree partition and the latter can be accessed via queries such as
``does $v$ belong to $L_i$'' or ``does $v$ belong to $H_{k,\ell}$.''
There is also a cost associated with each such query (see Definition \ref{def:degree-oracle} for details).


With the definition of degree partitions and 
degree oracles, our proof of Lemma~\ref{lem:up-given}
proceeds in the following two steps.
First we present in Lemma~\ref{lem:deg-oracle-access} an algorithm $\TwoApprox^*$~that achieves the same
goal as $\TwoApprox$, namely (\ref{output}) in Lemma~\ref{lem:up-given} with high probability.
The difference, however, is that $\TwoApprox^*$ is given access to not only~an~independent set oracle 
but also an
$(\ol{m},\eps)$-degree oracle. 
Next, we show in Lemma~\ref{lem:oracle-implementation} that an $(\ol{m},\eps)$-degree oracle 
can be implemented efficiently using 
access to the independent set oracle. 
This allows~us to {convert}
$\TwoApprox^*$ into $\TwoApprox$ with a similar performance guarantee, and Lemma \ref{lem:up-given} follows directly from Lemma~\ref{lem:deg-oracle-access} and Lemma~\ref{lem:oracle-implementation}.

We start with the definition of $(\ol{m},\eps)$-degree partitions:

\begin{definition}\label{def:degree-part}
	Let $G=([n],E)$ be a graph. 
	An \emph{$(\ol{m},\eps)$-degree partition} of $G$ is a partition 
	$$P=\Big(L_i,H_{k,\ell}: i\in [0:s], k\in [s+1:\beta]\ \text{and}\ \ell\in [0:\tau]\Big)$$
	of its vertex set $[n]$ \emph{(}so the sets in $P$ are disjoint and their union is $[n]$\emph{)} such that
	\begin{flushleft}\begin{enumerate}
			\item Let $L=\cup_i L_i$ and $H=\cup_{k,\ell} H_{k,\ell}$ \emph{(}so we have $L\cup H=[n]$\emph{)}.
			Every vertex $u\in L$ satisfies\\ $\deg(u)\le \alpha^{s+1}$ and 
			every vertex $u\in H$ satisfies $\deg(u)\ge \alpha^{s}$.
			
			\item Every vertex $u\in L_0$ satisfies $\deg(u,L)=0$ and every vertex $u\in L_i$, $i\in [s]$,
			satisfies
			\begin{equation}\label{exp2}
			\alpha^{i-1}\le \deg(u,L)\le \alpha^{i+1}.
			\end{equation}
			\item Let $H_k=\cup_\ell H_{k,\ell}$ for each $k\in [s+1:\beta]$.
			Then every vertex $u\in H_k$ satisfies 
			\begin{equation}\label{exp1}
			\alpha^{k-1}\le \deg(u)\le  \alpha^{k+1}.
			\end{equation}
			Moreover, every vertex $u\in H_{k,\ell}$ for some $\ell\in [0:\tau-1]$ satisfies 
			\begin{equation}\label{exp3}
			\alpha^{k-\ell-1}\le \deg(u,L)\le \alpha^{k-\ell+1} 
			\end{equation}
			and every $u\in H_{k,\tau}$ satisfies
			$ 
			\deg(u,L)\le \alpha^{k-\tau+1}.
			$
	\end{enumerate}\end{flushleft}
\end{definition}\vspace{-0.15cm}
\begin{remark}
	It is worth pointing out that intervals used in \emph{(\ref{exp2})}, \emph{(\ref{exp1})}, {and \emph{(\ref{exp3})}} are not disjoint
	\emph{(}and so~are the~conditions on $\deg(u)$ in the first item\emph{)}.
	As a result, such partitions are not unique for a given~$G$ in general. For example, a vertex with degree 
	between $\alpha^s$ and $\alpha^{s+1}$
	can lie in either $L$ or~$H$.
\end{remark}

Next we define $(\ol{m},\eps)$-degree oracles:

\begin{definition}\label{def:degree-oracle}
	Let $G=([n],E)$ be an undirected graph.
	An \emph{$(\ol{m},\eps)$-degree oracle} $\DO=(\DO_{\text{low}},\DO_{\text{high}})$ of $G$~contains an underlying $(\ol{m},\eps)$-degree 
	partition $P=(L_i,H_{k,\ell}:i,k,\ell)$ of $G$ and can be accessed via
	two maps $\DO_{\text{high}}:[n]\times [s+1:\beta]\times [0:\tau]\rightarrow \{0,1\}$ and $\DO_{\text{low}}:[n]\times [0:s]
	\rightarrow \{0,1\}$, where
	\begin{flushleft}\begin{enumerate}
			\item For every vertex $u\in [n]$, $\DO_{\text{high}}(u,k,\ell)=1$ if $u\in H_{k,\ell}$
			and $\DO_{\text{high}}(u,k,\ell)=0$ otherwise.\vspace{-0.1cm}
			\item For every vertex $u\in [n]$, $\DO_{\text{low}}(u,i)=1$ if $u\in L_i$
			and $0$ otherwise.
	\end{enumerate}\end{flushleft}
	The \emph{cost} of each query on $\DO_{\text{high}}$ is $1$ and the cost of each query $\DO_{\text{low}}(u,i)$
	is $\alpha^{s-i}$.
\end{definition}

We will be interested in algorithms that have access to both the independent set oracle $\IS$
and an $(\ol{m},\eps)$-degree oracle $\DO$ of a graph $G=([n],E)$.
For such an algorithm $\ALG^*$ (for clarity we always use $*$ to mark algorithms that 
have access to such a pair of oracles), we are interested in its \emph{total 
	cost}. The cost of each query on the independent set oracle 
is $1$, and the cost of each query on the degree oracle is specified in Definition \ref{def:degree-oracle}. {The total cost of an algorithm is the sum of the costs of individual queries.}

We are ready to state Lemma \ref{lem:deg-oracle-access} and Lemma 
\ref{lem:oracle-implementation} which together imply Lemma \ref{lem:up-given}.

\begin{lemma}[Estimation with degree oracles]\label{lem:deg-oracle-access}
	There is a randomized algorithm, $\emph{\texttt{Estimate-With-}}$ $\emph{\texttt{-Advice}}^* (\epsilon,n,\ol{m},G)$, that takes four inputs: an accuracy parameter $\eps\in (0,1)$,
	two positive integers $n$ and $\ol{m}$, and access to both the independent~set oracle $\IS$ 
	and an $(\ol{m},\eps)$-degree oracle $\DO$ of a graph $G=([n],E)$ with $1\le m=|E|\le \ol{m}$.
	Its worst-case total cost  
	is $(n/\sqrt{\ol{m}})\cdot \poly(\log n,1/\eps)$ and with probability at least $1-1/n^2$,   it
	returns $\hat{m} $ satisfying
	\begin{align} \label{eq:guarantees}
		(1-5\eps)m-O\left(\frac{\eps\ol{m}}{\log n}\right)
		&\leq \hat{m} \leq (1+\eps)m.
	\end{align}
\end{lemma}


We point out that, because $(\ol{m},\eps)$-degree partitions are \emph{not} unique,
${\TwoApprox}^*$ in Lemma \ref{lem:deg-oracle-access} needs to work with an $(\ol{m},\eps)$-degree oracle with \emph{any} underlying $(\ol{m},\eps)$-degree
partition (as long as it satisfies Definition \ref{def:degree-part}).
Lemma \ref{lem:oracle-implementation} below says that one can simulate a degree oracle efficiently using the independent set oracle.

\def\simD{\texttt{Sim-D}} \def\rr{\mathbf{r}}
\def\simDh{\simD_{\texttt{high}}}
\def\simDl{\simD_{\texttt{low}}}

\begin{lemma}[Simulation of degree oracles]\label{lem:oracle-implementation}
	Let $\eps\in (0,1)$ and $n,\ol{m}$ be positive integers. 
	There~are a  positive integer $q=q(\eps, n,\ol{m})$ and a pair of \emph{deterministic} algorithms \emph{$\simDl$}
	and \emph{$\simDh$},~where
	\emph{$\simDl$}$\hspace{0.04cm}(v,i,G,r)$ takes as input a vertex $v\in [n]$,
	$i\in [0:s]$, access to the independent set~\mbox{oracle}~of a graph $G=([n],E)$ with $1\le |E|\le \ol{m}$,
	and a string $r\in \{0,1\}^q$;
	\emph{$\simDh$}$\hspace{0.04cm}(v,k,\ell,G,r)$ takes~the same inputs but has 
	$i$ replaced by $k\in [s+1:\beta]$ and $\ell\in [0:\tau]$.
	Both algorithms output a value in $\{0,1\}$ and together have the following performance guarantee:
	\begin{flushleft}\begin{enumerate}
			\item \emph{$\simDl$}$\hspace{0.04cm}(v,i,G,r)$ makes $\alpha^{s-i}\cdot \poly(\log n,1/\eps)$ queries to $\IS$ and \emph{$\simDh$}$\hspace{0.04cm}(v,k,\ell,G,r)$ makes $\poly(\log n,1/\eps)$ queries to $\IS$.
			\item Given any graph $G$ with $1\le |E|\le \ol{m}$, when $\rr\sim \{0,1\}^q$ is drawn uniformly at random,
			\emph{$\simDl$}$\hspace{0.04cm}(v,i,G,\rr)$ viewed as a map from $[n]\times [0:s]\rightarrow \{0,1\}$
			and \emph{$\simDh$}$\hspace{0.04cm}(v,k,\ell,G,r)$ viewed as a map from $[n]\times [s+1:\beta]\times [0:\tau]\rightarrow \{0,1\}$ together form an $(\ol{m},\eps)$-degree\\ oracle of $G$ with probability at least $1-1/n^2$
			\emph{(}over the randomness of $\rr$\emph{)}.   
	\end{enumerate}\end{flushleft}
\end{lemma}

We use Lemma \ref{lem:deg-oracle-access} and \ref{lem:oracle-implementation}
to prove Lemma \ref{lem:up-given}.\medskip  

\begin{proofof}{Lemma \ref{lem:up-given} Assuming Lemma \ref{lem:deg-oracle-access} and \ref{lem:oracle-implementation}}
	The algorithm $\TwoApprox$ $\hspace{0.04cm}(\eps, n,$ $\ol{m},G)$ 
	draws a string $\rr\sim\{0,1\}^q$ uniformly at random, where $q=q(\eps,n,\ol{m})$ as in Lemma \ref{lem:oracle-implementation}, and  simulates $\TwoApprox^*$.
	When the latter makes a query on its given degree oracle, 
	$\TwoApprox$ runs either $\simD_{\texttt{low}}$ or $\simD_{\texttt{high}}$ using $\rr$
	and uses its output to continue the simulation of $\TwoApprox$.
	The query complexity of $\TwoApprox$~can~be bounded using the total cost of $\TwoApprox^*$ and 
	complexity of $\simD_{\texttt{low}}$ and $\simD_{\texttt{high}}$.
	The error probability of $\TwoApprox$ is at most $1/n^2$ (for the probability that $\rr$ fails 
	to produce an $(\ol{m},\eps)$-degree oracle) plus $1/n^2$ (for the error probability of $\TwoApprox^*$),
	which is smaller than $1/n$.
	This finishes the proof of Lemma \ref{lem:up-given}. 
\end{proofof}\medskip\vspace{0.06cm}

We prove Lemma~\ref{lem:deg-oracle-access} in the rest of this section and then prove Lemma~\ref{lem:oracle-implementation} in Section \ref{sec:check-degrees}.

\def\EstL{\texttt{Estimate-$L_i$}}
\def\EstH{\texttt{Estimate-$H_k$}}
\def\EstHH{\texttt{Estimate-$H_{k,\ell}$}}

\subsection{Estimation of $|L_i|$  and $|H_{k,\ell}|$. }\label{sec:estimate}

Let $G=([n],E)$ be the input graph with  $1\le m=|E|\le \ol{m}$.
We are given access to the independent set oracle $\IS$ and an 
$(\ol{m},\eps)$-degree oracle $\DO=(\DO_{\text{low}},\DO_{\text{high}})$ of $G$, where we use
$P=(L_i,H_{k,\ell}:i,k,\ell)$ to denote the degree partition underlying the degree oracle $\DO$. 
To obtain a good estimation of $|E|$, it suffices to obtain good estimations
of cardinalities of $L_i$'s and $H_{k,\ell}$'s (the latter would also lead to good estimations
of $|H_k|$; recall that $H_k=\cup_\ell H_{k,\ell}$).
Roughly speaking, estimations of $|L_i|$'s allow us to approximately count the number of 
edges in the subgraph induced by $L$;
estimations of $|H_k|$'s allow us to approximately count the total degree of 
vertices in $H$;
estimations of $|H_{k,\ell}|$'s allow us to approximately count the number of edges between $L$ and $H$.

We describe two subroutines for estimating $|L_i|$ and $|H_{k,\ell}|$ 
in Lemma \ref{lem:lowestimation} and \ref{lem:highestimation}, respectively, 
and then use them to prove Lemma \ref{lem:deg-oracle-access}.


\begin{lemma}[Estimation of $|L_i|$]\label{lem:lowestimation}
	Let $\eps\in (0,1)$ and $\ol{m}$ be a positive integer.
	There is a randomized algorithm that runs on graphs $G=({[n]},E)$ 
	with $1\le |E|\le \ol{m}$ 
	via access to the independent set oracle and an $(\ol{m},\eps)$-degree oracle of $G$ with
	an underlying $(\ol{m},\eps)$-degree partition $P=(L_i,H_{k,\ell}:i,k,\ell)$.
	It has total cost $ (n/\sqrt{\ol{m}}) \cdot \poly(\log n, 1/\eps)$
	and returns a number $\kappa_i$ for each $i\in [0:s]$ satisfying
	\begin{equation}\label{hehe1}
	|L_i|- \frac{\eps^2\hspace{0.03cm} \ol{m}}{\alpha^i \log^2n}\le \kappa_i\le |L_i| 
	\end{equation}
	with probability at least $1 - 1/n^3$.
\end{lemma}
\begin{proof}
	Fix an $i\in [0:s]$ and let $c_i=|L_i|/n$. We show how to compute $\kappa_i$.
	If
	\begin{equation}\label{hehe3}
	\frac{\eps^2\hspace{0.03cm}\ol{m}}{\alpha^i\log^2 n}\ge n,
	\end{equation}
	then we can set $\kappa_i=0$ and it satisfies (\ref{hehe1}) trivially.
	So we assume below that the inequality above does not hold.
	To estimate $c_i$  we draw 
	(the equation uses the assumption that (\ref{hehe3}) does not hold)
	$$
	\left\lceil\frac{n\alpha^i\log^5 n}{\eps^5\hspace{0.03cm}\ol{m}}\right\rceil= O\left(\frac{n\alpha^i\log^5 n}{\eps^5\hspace{0.03cm}\ol{m}}\right)
	$$
	vertices uniformly at random from $[n]$ (with replacements).
	For each vertex sampled, we query the degree oracle with a cost of $\alpha^{s-i}=\Theta(\sqrt{\ol{m}}/\alpha^i)$ to tell
	if it belongs to $L_i$.
	{The fraction of times that a vertex sampled belongs to $L_i$} gives us an empirical estimate $\hat{c}_i$ of $c_i$ and it follows from Chernoff bound 
	(using $c_in\cdot \alpha^{i-1}= |L_i|\cdot \alpha^{i-1}\le {2} |E|\le {2} \ol{m}$) that
	$$
	|\hat{c}_i-c_i|\le \frac{\eps^2\hspace{0.03cm}\ol{m}}{2\alpha^i n \log^2 n} $$
	with probability at least $1-\eps/n^4$.
	Setting $\kappa_i$ to be
	$$
	\kappa_i=\left(\hat{c}_i-\frac{\eps^2\ol{m}}{2\alpha^i n\log^2 n}\right)n
	$$
	would satisfy (\ref{hehe1}). The total cost for obtaining $\kappa_i$ is
	$(n/\sqrt{\ol{m}})\cdot \poly(\log n,1/\eps)$. The algorithm works on
	each $i$ and succeeds with probability at least $1-(s+1)\eps/n^4\ge 1-1/n^3$ by a 
	union bound.
\end{proof}

\begin{lemma}[Estimation of $|H_{k,\ell}|$]\label{lem:highestimation}
	Let $\eps\in (0,1)$ and $\ol{m}$ be a positive integer.
	There is a randomized algorithm that runs on $G=({[n]},E)$ with $1\le |E|\le \ol{m}$ 
	via access to the independent set~oracle and an $(\ol{m},\eps)$-degree oracle of $G$ with
	an underlying degree partition $P=(L_i,H_{k,\ell}:i,k,\ell)$.
	It has total cost $(n/\sqrt{\ol{m}}) \cdot \poly(\log n, 1/\eps)$
	and returns $\gamma_{k,\ell}$ for each $k\in [s+1:\beta]$ and $\ell\in [0:\tau]$ satisfying
	\begin{equation}\label{hehe2}
	\frac{|H_{k,\ell}|}{(1+\eps)^4}-O\left(\frac{\eps^4\hspace{0.03cm}\ol{m}}{\alpha^{k}\log^3 n}\right)\le \gamma_{k,\ell}\le |H_{k,\ell}|
	\end{equation}
	with probability at least $1 - 1/n^3$.
\end{lemma}


We delay the proof of Lemma \ref{lem:highestimation} to Section \ref{sec:finalproof}
but first use it to prove Lemma \ref{lem:deg-oracle-access}\medskip\vspace{0.06cm}

\begin{proofof}{Lemma~\ref{lem:deg-oracle-access} assuming Lemma  \ref{lem:highestimation}}
	Given $G$ and $P$, we let $m_1$, $m_2$ and $m_3$ denote
	$$
	m_1=\sum_{u\in L} \deg(u,L),\quad m_2=\sum_{u\in H} \deg(u)\quad \text{and}\quad
	m_3=\sum_{u\in H} \deg(u,L).
	$$
	Then we have $m=|E|=(m_1+m_2+m_3)/2$. The algorithm $\TwoApprox^*$ simply~runs the subroutines described in Lemma \ref{lem:lowestimation} and \ref{lem:highestimation} 
	to obtain $\kappa_i$'s and $\gamma_{k,\ell}$'s.
	Letting $\gamma_k=\sum_{{\ell\in[0:\tau]}} \gamma_{k,\ell}$, 
	it then outputs $\hat{m}=(\hat{m}_1+\hat{m}_2+\hat{m}_3)/2$, where
	$$
	\hat{m}_1= \sum_{i\in [s]} \kappa_i\cdot \alpha^i,\qquad
	\hat{m}_2=\sum_{k\in [s+1:\beta]} \gamma_k\cdot \alpha^k\qquad \text{and}\qquad
	\hat{m}_3=\sum_{\substack{k\in [s+1:\beta]\\ \ell\in [0:\tau-1]}} \gamma_{k,\ell}\cdot \alpha^{k-\ell}.
	$$
	Assuming that $\kappa_i$'s satisfy (\ref{hehe1}) and $\gamma_{k,\ell}$'s satisfy (\ref{hehe2})
	(which hold with probability at least $1-2/n^{{3}}$ 
	by Lemma \ref{lem:lowestimation} and Lemma  \ref{lem:highestimation}),
	we show in the rest of the proof that $\hat{m}$ satisfies (\ref{eq:guarantees}).
	This finishes the proof of the lemma since
	the worst-case total cost of $\TwoApprox^*$ can be bounded using Lemma \ref{lem:lowestimation} and Lemma \ref{lem:highestimation}.  
	
	First for  $m_1$, we have from (\ref{hehe1}) and the definition of
	$(\ol{m},\eps)$-degree partitions that
	$$
	\frac{m_1}{1+\eps}-O\left(\frac{\eps\hspace{0.03cm}\ol{m}}{\log n}\right)\le \sum_{{i \in [s]}} |L_i|\cdot \alpha^i-{O\left(\frac{\eps^2\hspace{0.03cm}\ol{m}}{\alpha^i\log^2 n}\right) s \alpha^i}
	\le \hat{m}_1\le \sum_{{i\in[s]}} |L_i|\cdot \alpha^i \le  (1+\eps) m_1 .
	$$
	Next, from (\ref{hehe2}) combined with the fact that $|H_k|=\sum_{{\ell\in[0:\tau]}} |H_{k,\ell}|$ and  
	{$\tau=\Theta(\log(\log n/\eps)/\eps)=O(\log \log n/\eps^2)$} 
	we have that
	$$
	\frac{|H_k|}{(1+\eps)^4}-O\left( \frac{\eps^2 \hspace{0.03cm} \ol{m}}{\alpha^k\log^2 n}\right)\le 
	\frac{|H_k|}{(1+\eps)^4}-O\left(\tau \cdot \frac{\eps^4 \hspace{0.03cm} \ol{m}}{\alpha^k\log^3 n}\right)\le \gamma_k=\sum_{{\ell\in[0:\tau]}}\gamma_{k,\ell} \le |H_k|.
	$$
	As a result, we have from the definition of $(\ol{m},\eps)$-degree partitions that
	\begin{align*}
		\frac{m_2}{(1+\eps)^5}-O\left(\frac{\eps\hspace{0.03cm}\ol{m}}{\log n}\right) &\le \sum_{{k\in[s+1:\beta]}} \frac{|H_k|}{(1+\eps)^4}\cdot \alpha^k-{O\left(\frac{\eps^2\ol{m}}{\alpha^k\log^2n}\right) {(\beta-s-1)} \alpha^k} \\
		&\le  \hat{m}_2\le \sum_{k\in[s+1:\beta]} |H_k|\cdot \alpha^k\le (1+\eps) m_2.
	\end{align*}
	Finally the following upper bound for $\hat{m}_3$ follows from (\ref{hehe2}):
	$$
	\hat{m}_3\le \sum_{\substack{k\in [s+1:\beta]\\ \ell\in [0:\tau-1]}} |H_{k,\ell}|\cdot \alpha^{k-\ell}
	\le (1+\eps) m_3.
	$$
	For a lower bound  note that ${\sum_{k \in [s+1:\beta]}} \alpha^{k-1} |H_k|\le 2|E|\le 2\ol{m}$.\vspace{0.1cm} 
	Together with {$\alpha^\tau =\Theta(\log n/\eps)$} we have
	$$
	m_3\le \sum_{\substack{k\in [s+1:\beta]\\ \ell\in [0:\tau-1]}} |H_{k,\ell}|\cdot \alpha^{k-\ell+1}
	+ \sum_{{k \in [s+1:\beta]}} |H_k|\cdot \alpha^{k-\tau+1}
	\le \sum_{\substack{k\in [s+1:\beta]\\ \ell\in [0:\tau-1]}} |H_{k,\ell}|\cdot \alpha^{k-\ell+1}
	+ O\left(\frac{\eps \hspace{0.03cm}\ol{m}}{\log n}\right).\vspace{0.1cm}
	$$
	As a result, we have from (\ref{hehe2}) that 
	\begin{align*}
		\hat{m}_3\ge \sum_{\substack{k\in [s+1:\beta]\\ \ell\in [0:\tau-1]}} \frac{|H_{k,\ell}|}{(1+\eps)^4} \cdot \alpha^{k-\ell}
		-O\left(\frac{\eps\hspace{0.03cm}\ol{m}}{\log n}\right)
		&\ge \frac{1}{(1+\eps)^5}\cdot \left( {m_3} -O\left(\frac{\eps\hspace{0.03cm}\ol{m}}{\log n}\right)\right)-O\left(\frac{\eps\hspace{0.03cm}\ol{m}}{\log n}\right).\\[-1.8ex]
		&\ge \frac{m_3}{(1+\eps)^5}-O\left(\frac{\eps\hspace{0.03cm}\ol{m}}{\log n}\right).
	\end{align*}
	It follows that 
	$$
	(1-5\eps) m -O\left(\frac{\eps\hspace{0.03cm}\ol{m}}{\log n}\right)\le 
	\frac{m}{(1+\eps)^5}-O\left(\frac{\eps\hspace{0.03cm}\ol{m}}{\log n}\right)\le \hat{m}\le (1+\eps) m.
	$$
	This finishes the proof of the lemma.
\end{proofof}

\subsection{Proof of Lemma \ref{lem:highestimation}}\label{sec:finalproof}

In this subsection we will prove Lemma~\ref{lem:highestimation}. Specifically, 
fixing any $k\in [s+1:\beta]$ and $\ell\in [0:\tau]$ 
we will design a procedure to approximate the size of $H_{k,\ell}$. 
Our procedure $\HiDegBucket^*$ for this purpose uses a subroutine called $\HiDegEvent^*$.
Its performance guarantee  is proved in the following lemma:

\medskip
\begin{figure}[t!]
	\begin{framed}
		\noindent Subroutine $\HiDegEvent^*\hspace{0.05cm}(k,\ell,{\eta},G)$
		\begin{flushleft}
			\noindent {\bf Input:} 
			Integers $k\in [s+1:\beta]$ and $\ell\in [0:\tau]$, a parameter ${\eta \in [0,1]}$ satisfying (\ref{para1}), 
			and access to both the independent set oracle $\IS$ and an $(\ol{m},\eps)$-degree
			oracle $\DO$ (with underlying degree partition $P=(L_i,H_{k,\ell}:i,k,\ell)$) of a  graph $G=([n],E)$ with $1\le m=|E|\le \ol{m}$.\\ 
			\noindent{\bf Output:} Either ``few'' or  ``many.''
			\begin{enumerate}
				\item Initialize a counter $c\leftarrow 0$, and repeat the following $N$ times: 
				\begin{enumerate}
					\item Sample an $\bS\subseteq [n]$ where each vertex is included with probability $p$  independently.
					\item Sample an $\bT\subseteq [n]$ where each vertex is included with probability $q$ independently. 					
					\item If \ignore{$\bS$ and }$\bT$ {is an}\ignore{are} independent set and $\bS\cup\bT$ is not an independent set
					(via $\IS$)\vspace{0.12cm}
					\begin{enumerate}
						
						\item Run $\BinarySearch\hspace{0.04cm}(n,G,\bS\cup \bT,\eps^7/n^2)$ to find an edge $(u,v)$ in $\bS\cup \bT$.\vspace{0.06cm}
						\item Query $\DO_{\text{high}}(u,k,\ell)$ and $\DO_{\text{high}}(v,k,\ell)$.\vspace{0.06cm} 
						\item \ignore{If one of $\{u,v\}$ lies in both $H_{k,\ell}$ and $\bS$ 
							and the other lies in $\bT$,\footnotemark\ then $c\leftarrow c+1$.} {If $u^* \in \{ u, v\}$ lies in $\bS$ and $H_{k, \ell}$, and $\{u^*\} \cup \bT$ is not an independent set (via $\IS$), let $c \leftarrow c+1$.}
					\end{enumerate}
				\end{enumerate}
				\item If $ {c\ge h}$, {\bf return} ``many;'' otherwise {\bf return} ``few.'' 
			\end{enumerate}
		\end{flushleft}\vskip -0.14in
	\end{framed}\vspace{-0.2cm}
	\caption{Description of the $\HiDegEvent^*$ subroutine.} \label{fig:HiDegEvent}
\end{figure}

\ignore{\footnotetext{In more details, if either (1) $u\in H_{k,\ell}$, 
		$u\in \bS$ and $v\in \bT$ or (2) $v\in H_{k,\ell}$,
		$v\in \bS$ and $u\in \bT$ we increment the~counter $c$.
		Note that in both cases it is fine for $u$ or $v$
		to lie in $\bS\cap \bT$.}}

\begin{lemma} \label{lem:HighDegEventLemma} 
	There is a randomized algorithm $\emph{\HiDegEvent}^*\hspace{0.04cm}(k, \ell, {\eta},G)$ that takes 
	the following inputs\footnote{For convenience we skip $n,\ol{m}$ and $\eps$ as
		inputs of $\HiDegEvent^*$ and $\HiDegBucket^*$.}: integers $k\in [s+1:\beta]$ and 
	$\ell\in [0:\tau]$, a parameter ${\eta} \in [0,1]$ satisfying\footnote{Note that the left hand side of (\ref{para1}) 
		is smaller than $1$ given that $\alpha^k>\alpha^s=\Theta(\sqrt{\ol{m}})$ and $\smash{\ol{m}\le {n\choose 2}}$.}   
	\begin{equation}\label{para1}\frac{\eps^4\ol{m}}{\alpha^k n\log^3 n}\le {\eta}\le 1,\end{equation}
	and access to the independent set oracle $\IS$ and an $(\ol{m},\eps)$-degree oracle $\DO$
	of $G=([n],E)$ satisfying $1\le m=|E|\le \ol{m}$.
	\hspace{-0.04cm}The algorithm $\emph{\HiDegEvent}^* $
	has a total cost of $(n/\sqrt{\ol{m}})\cdot \poly(\log n,1/\eps)$ 
	and has the following performance guarantee. Let
	$P=(L_i,H_{k,\ell}:i,k,\ell)$ denote the degree partition of the given degree oracle $\DO$ of $G$. Then
	\begin{enumerate}
		\item If $|H_{k,\ell}| \le {\eta} n$, then the algorithm  outputs ``few'' with probability at least
		$1-\eps^4/n^4$;\vspace{-0.1cm}
		\item If $|H_{k,\ell}|\ge \alpha^3{\eta} n$, then the algorithm outputs ``many''
		with probability at least $1-\eps^4/n^4$.\vspace{0.1cm}
	\end{enumerate}
\end{lemma}

\begin{proof} 
	We describe $\HiDegEvent^*$ in Figure~\ref{fig:HiDegEvent} with the following four parameters
	(one can check that $p<1$ using the condition on {$\eta \in [0,1]$} \ignore{$\gamma$}
	in (\ref{para1})):
	\begin{equation}\label{parameters}
	N=\frac{n\log^7n}{\epsilon^{9}\sqrt{\mbar} },\quad
	h=\left(1+\frac{\eps}{4}\right)\frac{\log^2 n}{\eps^3},\quad p=\frac{\eps^5 \sqrt{ \mbar}}{{\eta} n^2 \log^4 n}\quad\text{and}\quad q=\frac{\eps}{\alpha^{k+1} \log n}.
	\end{equation}
	
	Suppose that $|H_{k,\ell}|\le {\eta} n$, and consider the probability that the counter is incremented at any specific iteration of $\HiDegEvent^*$. 
	Note that a necessary condition for this to happen is that there is a vertex $u^*\in H_{k,\ell}$
	that is included in $\bS$ and $u^*$ has a neighbor in $\bT$ {so that step 1(c)iii increments the counter $c$}.
	Thus we have
	$$
	\Prx_{\bS, \bT}\big[c \text{ is incremented}\big]\le \sum_{u\in H_{k,\ell}} \Prx_{\bS}\big[u\in \bS\big]
	\cdot \Prx_{\bT}\big[\bT\ \text{contains a neighbor of $u$}\big].
	$$
	Given that every vertex $u\in H_{k,\ell}$ has degree at most $\alpha^{k+1}$. We have
	$$
	\Prx_{\bT}\big[\bT\ \text{contains a neighbor of $u$}\big]\le \alpha^{k+1}q
	= \frac{ \eps}{\log n} .
	$$
	As a result, for the case when $|H_{k,\ell}|\le {\eta} n$ we have 
	$$
	\Prx_{\bS, \bT}\big[c \text{ is incremented}\big]\le  \frac{{\eta} n p\hspace{0.02cm}\eps}{\log n}.
	$$
	
	
	Next we consider the case of $|H_{k,\ell}|\ge \alpha^3{\eta} n$. 
	A \emph{sufficient} condition for the counter to increment is that there is a vertex $u^* \in H_{k,\ell}$
	such that (1) $u^* \in \bS$ \ignore{and $u^* \notin \bT$}, (2) one of the neighbors of $u^*$ lies in $\bT$, (3) $(\bS\setminus \{u^*\})\cup \bT$
	is an independent set, and (4) $\BinarySearch$ does not fail. {Suppose these occur for a sample of $\bS$ and $\bT$ in step 1(a) and 1(b). Then, $\bT \subset (\bS \setminus \{u^*\}) \cup \bT$ must be an independent set by (3), and $\bS \cup \bT$ is not an independent set by (1) and (2). This means step 1(c) enters lines (i), (ii) and (iii). By (3) and (4), $\BinarySearch(m, G, \bS \cup \bT, \eps^7/n^2)$ outputs an edge $(u^*, v)$ since all edges in $\bS \cup \bT$ are adjacent to $u^*$; hence, (ii) executes $\DO_{\text{high}}(u^*, k, \ell)$ and notices $u^*$ lies in $\bS$ and $H_{k,\ell}$. Finally by (2), $\{ u^*\} \cup \bT$ is not an independent set in (iii) and the counter is incremented.}\ignore{\enote{Being very explicit about how the event (1),(2), (3) and (4) in lines (a) and (b) give rise to an execution which increments the counter in (iii).}}
	
	{We first show that the events (1), (2), and (3) are disjoint for different vertices $u \in H_{k, \ell}$. Suppose for contradiction that $u_1, u_2 \in H_{k, \ell}$ satisfy events (1), (2), and (3). Then, by (3), $(\bS \setminus \{u_1\}) \cup \bT$ and $(\bS \setminus \{u_2\}) \cup \bT$ are independent sets, which means that $(u_1, u_2)$ is the only edge in $\bS \cup \bT$. This implies by applying (2) to $u_1$ that $u_2 \in \bT$, and similarly $u_1 \in \bT$ by applying (2) to $u_2$. Thus, there is an edge in $(\bS \setminus \{u_1\}) \cup \bT$, a contradiction.}
	\ignore{Note that, by adding the condition of $u\notin \bT$, the above events over different vertices $u\in H_{k,\ell}$ are disjoint. }Thus,
	the probability for $c$ to increment is at least (the last term accounts for $\BinarySearch$)
	$$
	\sum_{u\in H_{k,\ell}} \Prx_{\bS, \bT} \big[u\in \bS, \ignore{\hspace{0.03cm}u\notin \bT, }
	\hspace{0.03cm}\bT\ \text{contains a neighbor of $u$ and
		$(\bS\setminus \{u\})\cup \bT$ is an IS}\big]-(\eps^7/n^2).
	$$
	Let $\bS'$ be the set drawn by including each vertex in $[n]\setminus \{u\}$ with probability $p$ independently,
	and let $\bT'$ be the set drawn similarly {from $[n] \setminus \{u\}$} using $q$.
	Then the probability in the sum above can be written as
	$$
	p(1-q)\cdot \Pr_{\bS',\bT'}\big[\bT'\ \text{contains a neighbor of $u$ and $\bS'\cup \bT'$ is an IS}\big].
	$$
	On the one hand, the probability that $\bT'$ contains a neighbor of $u$ is at least 
	$$
	\alpha^{k-1}q-{\alpha^{k-1} \choose 2}q^2\ge \left(\frac{1}{\alpha^2}-O\left(\frac{\eps}{\log n}\right)\right)\cdot \frac{\eps}{\log n} 
	$$
	as $\alpha^{k+1} q=\eps/\log n$. 
	On the other hand, it follows from Lemma \ref{lem:indset}, (\ref{para1}) and $\alpha^k=O(n)$ that
	$$
	\Pr_{\bS',\bT'}\big[\text{$\bS'\cup \bT'$ is not an IS}\big]\le \ol{m}\cdot (p+q)^2\le \ol{m}\cdot 2(p^2+q^2)
	=O\left(\frac{\eps^2}{\log^2 n}
	\right).
	$$
	As a result, we have 
	\begin{align*}
		\Pr_{\bS',\bT'}\big[\bT'\ \text{contains a neighbor of $u$ and $\bS'\cup \bT'$ is an IS}\big]\ge 
		\left(\frac{1}{\alpha^2}-O\left(\frac{\eps}{\log n}\right)\right)\cdot \frac{\eps}{\log n}.
	\end{align*}
	So for the case when $|H_{k,\ell}|\ge \alpha^3{\eta} n$, we have 
	$$
	\Prx_{\bS, \bT}\big[c \text{ is incremented}\big]\ge \alpha^3{\eta} n\cdot p(1-q) \cdot \left(\frac{1}{\alpha^2}-O\left(\frac{\eps}{\log n}\right)\right)\cdot \frac{ \eps}{\log n}-\frac{\eps^7}{n^2}
	\ge (1+\eps/2)\cdot \frac{{\eta} n p\hspace{0.02cm}\eps}{\log n}.
	$$
	Plugging in the choices of $p$ and $N$, we have that
	$$
	N=\frac{\log n}{{\eta} n p \eps}\cdot \frac{\log^2 n}{\eps^4} .
	$$
	By a Chernoff bound the counter will distinguish the two cases with probability $1-\eps^3/n^4$.  
\end{proof}\medskip

Using the above lemma, we can estimate the sizes of the high degree buckets.\medskip
\begin{figure}[t!]
	\begin{framed}
		\noindent Procedure $\HiDegBucket^*\hspace{0.05cm}(k,\ell,G)$
		\begin{flushleft}
			\noindent {\bf Input:} Integers $k\in [s+1:\beta]$ and $\ell\in [0:\tau]$,   
			and access to both the independent set oracle $\IS$ and an $(\ol{m},\eps)$-degree
			oracle $\DO$ (with underlying degree partition $P=(L_i,H_{k,\ell}:i,k,\ell)$)\\ of an undirected graph $G=([n],E)$ with $1\le m=|E|\le \ol{m}$.\\
			\noindent{\bf Output:} An estimation $\gamma_{k,\ell}$ of $|H_{k,\ell}|$.
			
			\begin{enumerate}
				\item Let ${\eta}=1$ 
				\item While ${\eta}\ge {\eps^4\ol{m}}\big/({\alpha^k n\log^3 n})$, perform the following:
				\begin{enumerate}
					\item Run $\HiDegEvent^*(k,\ell,{\eta}, G)$
					\item If it outputs ``many,'' {\bf return} ${\eta} n$ as $\gamma_{k,\ell}$; Otherwise, set ${\eta}$ to be ${\eta}/\alpha$.
				\end{enumerate}
				\item			{\bf Return} $0$	
			\end{enumerate}
		\end{flushleft}\vskip -0.14in
	\end{framed}\vspace{-0.2cm}
	\caption{Description of the $\HiDegBucket^*$ procedure.} \label{fig:HiDegBucket}
\end{figure}

\begin{proofof}{Lemma~\ref{lem:highestimation}}
	The algorithm simply runs $\HiDegBucket^*\hspace{0.03cm}(k,\ell,G)$ for each $k {\in [s+1:\beta]}$ and $\ell {\in [0 : \tau]}$ to obtain $\gamma_{k,\ell}$. Its total cost can be bounded easily given that $\HiDegBucket^*$ only invokes $\HiDegEvent^*$
	at most $O(\log n/\eps)$ many times, and both $\beta$ and $\tau$ are $\wt{O}(\log n/\eps)$.
	
	Below we assume that every call to $\HiDegEvent^*(k,\ell,{\eta},G)$ satisfies the two conditions in Lemma 
	\ref{lem:HighDegEventLemma}, which happens with probability at least
	$$
	1-{\frac{\eps^4}{n^4}}\cdot \wt{O}\left(\frac{\log^3 n}{\eps^3}\right)>1-\frac{1}{n^3}.
	$$ \ignore {\enote{before, this was $O(\frac{\eps^3}{n^4})$, but we get exact values from the above lemma.}}
	We show that every $\gamma_{k,\ell}$ satisfies (\ref{hehe2}) and the lemma follows.
	
	Let $\gamma_{k,\ell}$ be the output of $\HiDegBucket^*(k,\ell,G)$.
	Considered two cases. First suppose line $3$ is reached so $\gamma_{k,\ell}=0$.
	Let $\hat{{\eta}}$ be the value of ${\eta}$ in the last call to $\HiDegEvent^*$.
	Then  
	$$
	\hat{{\eta}} {\leq \frac{\eps^4 \ol{m}}{\alpha^{k-1} n \log^3 n}} 
	$$
	and because every call to $\HiDegEvent^*$ returns a correct answer (``few'' in this case), 
	$$
	|H_{k,\ell}|\le \alpha^3 \hat{{\eta}} n=O\left(\frac{\eps^4\ol{m}}{\alpha^k \log^3 n}\right)
	$$
	so (\ref{hehe2}) holds trivially with $\gamma_{k,\ell}=0$.
	
	\medskip
	Next suppose that $\gamma_{k,\ell}=\hat{{\eta}}n$ since
	$\HiDegEvent^*\hspace{0.03cm}(k,\ell,\hat{{\eta}},G)$ outputted ``many'', and~the previous $\HiDegEvent^*\hspace{0.03cm}
	(k,\ell,\alpha \hat{{\eta}} ,G)$ outputted ``few.'' 
	Given the assumption that both invocations return correct answers, we get that 
	${\gamma_{k,\ell} =} \hat{{\eta}}n\le |H_{k,\ell}|$ {and $|H_{k,\ell}| \leq \alpha^4\hat{\eta} n = \alpha^4 \gamma_{k,\ell}$}, so  (\ref{hehe2}) follows.	
\end{proofof}

\newcommand{\CheckHLFrac}{\texttt{Check-H-L-Degree}}
\newcommand{\CheckHLBucket}{\texttt{Check-H-L-Bucket}}

\section{Simulation of Oracles} \label{sec:check-degrees}

We prove Lemma~\ref{lem:oracle-implementation} in this section. We show how to simulate access to an $(\ol{m}, \eps)$-degree oracle  by giving implementations of $\simDh$ and $\simDl$, which assume access to an independent set oracle. 
To simplify the presentation, we break the simulation into two steps.
In the first step, we introduce the notion of a \emph{high-low partition} and a \emph{high-low oracle}
in Section \ref{sec:highlow} and show how to simulate~a high-low oracle using access to an independent set
oracle.
In the second step, we show how to~simulate an $(\ol{m},\eps)$-degree oracle 
with access to both an independent set oracle and a high-low oracle. 


Throughout the section, let $\eps \in (0,1)$ be an accuracy parameter, $1\le \ol{m}\le {n\choose 2}$ 
and $G = ([n], E)$ be a graph where $1 \leq m = |E| \leq \ol{m}$. Recall $\alpha = 1+\eps$, $s$ is set according to (\ref{eq:def-s}), $\beta = \Theta((\log n)/\eps)$ is the smallest integer such that $\alpha^{\beta} \geq n$, and $\tau$ is the smallest integer such that $\alpha^{\tau} \geq \log^2 n / \eps$.  
For convenience we will fix $\eps$ and $\ol{m}$ and skip them as inputs of algorithms presented in this section.

\subsection{High-low partitions and oracles}\label{sec:highlow}

We start with the definition of high-low partitions and oracles.

\def\DOHL{\mathsf{D}_{\mathsf{HL}}}
\begin{definition}
	An $(\ol{m}, \eps)$-\emph{high-low partition} of $G=([n],E)$ is a partition $(H, L)$ of $[n]$ such that every vertex $u \in L$ satisfies $\deg(u) \leq \alpha^{s+1}$ and every vertex $u \in H$ satisfies $\deg(v) \geq \alpha^s$.
	
	An $(\ol{m}, \eps)$-\emph{high-low oracle} contains an  $(\ol{m}, \eps)$-high-low partition~$ (H, L)$~of $G$, and can be accessed via a map $\DOHL \colon [n] \to \{0,1\}$ such that $\DOHL(u) = 1$ if $u \in H$ and $\DOHL(u) = 0$ if $u \in L$. 
\end{definition}

{We remark (similarly to the case of $(\ol{m}, \eps)$-degree partitions in Definition~\ref{def:degree-part}) that $(\ol{m}, \eps)$-high-low partitions are not unique; in fact, a vertex $v$ with $\alpha^s \leq \deg(v) \leq \alpha^{s+1}$ may belong to either $H$ or $L$ in an $(\ol{m}, \eps)$-high-low partition $(H,L)$.} We show in the next lemma that query access to an $(\ol{m},\eps)$-high-low oracle $\DOHL$ can be
simulated very efficiently using an independent set oracle.

\def\HL{\texttt{High-Low}}

\begin{lemma}\label{lem:high-low-simul}
	There is a positive integer $q_\mathsf{HL} = q_\mathsf{HL}(\eps, n, \ol{m})$ and a deterministic algorithm 
	\emph{$\HL$} with the following performance guarantee.
	\emph{$\HL$}$\hspace{0.04cm}(u,G,r)$ takes three inputs: {a vertex} $u \in [n]$,
	access to an independent set oracle $\IS$ of $G=([n],E)$ with $1\le m=|E|\le \ol{m}$, and $r \in \{0,1\}^{q_\mathsf{HL}}$. The algorithm makes at most $\poly(\log n, 1/\eps)$ queries to $\IS$ and outputs a value in $\{0,1\}$. {With probability at least $1 - 1/n^3$ over the draw of $\boldr \sim \{0,1\}^{q_{\mathsf{HL}}}$, the function $\emph{\HL}(\cdot, G, \boldr) \colon [n] \to \{0,1\}$, is an $(\ol{m}, \eps)$-high-low oracle of $G$.}\ignore{\enote{reworded a bit to make it clear that the probability is only over the draw of $\boldr$.}}
\end{lemma}

Before giving the proof of Lemma~\ref{lem:high-low-simul}, we introduce the main subroutine, $\CheckHiDegree$, which will be used for $\HL$ as well as for later parts of this section.

\begin{figure}[t!]
	\begin{framed}
		\noindent Subroutine $\CheckHiDegree\hspace{0.05cm}(u, d, G)$
		\begin{flushleft}
			\noindent {\bf Input:} A vertex $u \in [n]$, {a parameter} {$d \geq \alpha^s$}, 
			and  access to an independent set oracle\\ $\IS$ of an undirected graph $G=([n],E)$
			with $1\le m=|E|\le \ol{m}$. \\
			{\bf Output:} Either ``low'' or ``high.''
			
			\begin{enumerate}
				\item Let $c$ be a counter, initially set to $0$. Repeat $t = \poly(\log n,1/\eps)$ many iterations:
				\begin{itemize}
					\item Sample $\bT \subseteq [n] \setminus \{ u \}$ by including each vertex independently with probability $\eps/(d\log n)$. Increment $c$ if $\bT$ is an independent set but $\bT \cup \{u\}$ in not.
				\end{itemize}
				\item If $c > (1+\eps/4)t \eps \big/ \log n$, output ``high;'' otherwise, output ``low.'' 
			\end{enumerate}
		\end{flushleft}\vskip -0.14in
	\end{framed}\vspace{-0.2cm}
	\caption{Description of the $\CheckHiDegree$ subroutine.} \label{fig:high-degree-check}
\end{figure}

\begin{lemma} \label{lem:CheckHiDegLemma}
	There is a randomized algorithm $\emph{\CheckHiDegree}\hspace{0.05cm}(u, d, G)$ which takes three inputs: a vertex $u \in [n]$, {a parameter}\ignore{\enote{$d$ doesn?t need to be an integer, also changing this in the subroutine description}} $d {\geq \alpha^{s}}$, 
	and access to~an independent set oracle
	of $G=([n],E)$ with $1\le |E|\le \ol{m}$.
	The algorithm makes at most $\poly(\log n,1/\eps)$ queries and satisfies the following two properties:
	\begin{itemize}
		\item If $\deg(u)\ge (1+\eps)d $, then {$\emph{\CheckHiDegree}(u, d, G)$} outputs ``high''
		with probability at least $1-{\eps^2/n^5}$.\vspace{-0.1cm}
		\item If $\deg(u)\le d$, then {$\emph{\CheckHiDegree}(u, d, G)$} outputs ``low''
		with probability at least $1-{\eps^2/n^5}$.
	\end{itemize}
\end{lemma}

\begin{proof}
	Suppose first $\deg(u)\ge (1+\eps)d$. Consider the probability over the draw of $\bT \subseteq [n] \setminus \{u\}$~that the counter $c$ is incremented at any particular iteration. We notice that if $\bT$ is an independent set containing a neighbor of $u$, the counter is incremented. Therefore,
	\begin{align*} 
		\hspace{-0.3cm}\Prx_{\bT}\big[ c\text{ is incremented}\big] &\geq
		\Prx_{\bT}\big[\bT \cap \Gamma(u) \neq \emptyset \big] - \Prx_{\bT}\big[\bT\text{ is not an independent set}\big] \\
		&\geq 1 - \left(1 - \frac{\eps}{d \log n}\right)^{(1+\eps) d} - O\left(\frac{\eps^2}{\log^2 n}\right) \geq 
		\frac{\eps(1+\eps)(1-o(\eps))}{\log n}\ge 
		\frac{\eps(1+\eps/2)}{\log n}.
	\end{align*}
	where we used Lemma~\ref{lem:indset} to say that $\bT$ is very likely to be an independent set. 
	On the other hand when $\deg(u) \leq d$, the probability that the counter is incremented is at most the probability that any neighbor of $u$ is included in $\bT$, so at most $\eps/\log n$. By a Chernoff bound, the counter $c$ at the end will be able to distinguish the two cases with probability at least $1 -{\eps^2/n^5}$. 
\end{proof}\medskip\vspace{0.06cm}

We now use Lemma \ref{lem:CheckHiDegLemma} to prove Lemma \ref{lem:high-low-simul}:

\begin{proofof}{Lemma~\ref{lem:high-low-simul}}
	Let $q_\mathsf{HL} = q_\mathsf{HL}(\eps, n, \ol{m})$ be a large enough integer so that $r \in \{0,1\}^{q_\mathsf{HL}}$ can store the randomness of calls to $\CheckHiDegree\hspace{0.04cm}(u, \alpha^{{s}}, G)$\ignore{\enote{This was previously set to call $\CheckHiDegree(u, \alpha^{s+1}, G)$, which I think its shifted.}} for every $u \in [n]$. More formally, if $\kappa$ is the number of random bits needed for
	each call to $\CheckHiDegree\hspace{0.04cm}(u, \alpha^{{s}}, G)$, then $q_\mathsf{HL}$ is set to be $n \cdot \kappa$. 
	By a union bound, with probability at least $1 - {\eps^2/n^4}$ over the draw of $\boldr \sim \{0,1\}^{q_\mathsf{HL}}$, all $n$ calls to $\CheckHiDegree(u, \alpha^{{s}}, G)$ 
	return a correct answer (i.e. no property in 
	Lemma~\ref{lem:CheckHiDegLemma} is violated). 
	We will refer to such a string $r$ as a \emph{good} string.  
	
	We now describe the implementation of $\HL(u,G,r)$ and 
	show that for every good string $r$, $\HL(\cdot,G,r)$ implements an $(\ol{m},\eps)$-high-low oracle. 
	When calling $\HL(u,G, r)$, it~just calls $\CheckHiDegree(u, \alpha^{{s}}, G)$ with 
	randomness taken from bits of $r$ allocated to this~call.
	Then $\HL(u,G,r)$ outputs $1$ if it outputs ``high,'' and $0$ if it outputs ``low.''
	It follows from Lemma \ref{lem:CheckHiDegLemma} that $\HL(u,G,r)$ makes $\poly(\log n,1/\eps)$ independent set queries.
	Moreover, when $r$ is a good string, $\HL(u,G,r)=1$ implies that $\deg(u)\ge \alpha^s$; 
	$\HL(u,G,r)=0$ implies that $\deg(u)\le \alpha^{s+1}$.
	This finishes the proof of Lemma \ref{lem:high-low-simul}.
\end{proofof}

\subsection{Implementation of a degree oracle using a high-low oracle}

Lemma~\ref{lem:oracle-implementation} follows from Lemma \ref{lem:high-low-simul}
and the next lemma which is almost identical to Lemma~\ref{lem:oracle-implementation},
except that the algorithms now have access to both an independent set oracle and 
a high-low oracle.

\begin{lemma} \label{lem:oracle-implementation2}
	There~{exists} a  positive integer $q_{*}=q_{*}(\eps, n,\ol{m})$ and two \emph{deterministic} algorithms \emph{$\simDl^*$}
	and \emph{$\simDh^*$},~where
	\emph{$\simDl^*$}$\hspace{0.04cm}(u,i,G,r)$ takes as input a vertex $u\in [n]$,
	an index $i\in [0:s]$, access to both an independent set~\mbox{oracle} $\IS$ and 
	an $(\ol{m},\eps)$-high-low oracle $\DOHL$ of an undirected graph $G=([n],E)$ with $1\le m=|E|\le \ol{m}$,
	and an $r\in \{0,1\}^{q_{*}}$;
	\emph{$\simDh^*$}$\hspace{0.04cm}(u,k,\ell,G,r)$ takes~the same inputs but has the index
	$i$ replaced by indices $k\in [s+1:\beta]$ and $\ell\in [0:\tau]$.
	Both algorithms output a value in $\{0,1\}$ and together~have the following performance guarantee:
	\begin{flushleft}\begin{enumerate}
			\item \emph{$\simDl^*$}$\hspace{0.04cm}(u,i,G,r)$ makes $\alpha^{s-i}\cdot \poly(\log n,1/\eps)$ queries and \emph{$\simDh^*$}$\hspace{0.04cm}(u,k,\ell,G,r)$ makes $\poly(\log n,1/\eps)$ queries to the two oracles $\IS$ and $\DOHL$.
			\item {With probability at least $1-1/n^3$ over $\rr\sim \{0,1\}^{q_{*}}$},
			\emph{$\simDl^*$}$\hspace{0.04cm}(u,i,G,\rr)$ viewed as a map from ${[n]}\times [0:s]\rightarrow \{0,1\}$\ignore{\enote{Before, this map was from vertices in $V$, so changed to $[n]$}}
			and \emph{$\simDh^*$}$\hspace{0.04cm}(u,k,\ell,G,\rr)$ viewed as a map from ${[n]}\times [s+1:\beta]\times [0:\tau]\rightarrow \{0,1\}$ form an $(\ol{m},\eps)$-degree oracle of $G$.
	\end{enumerate}\end{flushleft}
\end{lemma}

To prove Lemma \ref{lem:oracle-implementation2}, we need two procedures 
with properties summarized in the following two lemmas.
We delay their proofs but first use them to prove Lemma \ref{lem:oracle-implementation2}.

\begin{lemma}\label{lem:high-low-frac-lemma}
	There is a randomized algorithm $\emph{\CheckHLFrac}\hspace{0.05cm}{(u, k, \ell, G)}$\ignore{\enote{Changed the order of these to sync up with the functions in Lemma 4.4 -- for some reason, the vertex $u$ was the third input.}} which takes as input {a vertex} $u\in [n]$, two integers $k \in [s +1 : \beta]$ and $\ell\in [\tau]$,\ignore{\enote{Before, $\ell \in [0 : \tau]$, so removing the case $\ell = 0$, since this is always a default case} } and access to both an independent set oracle and
	an $(\ol{m},\eps)$-high-low oracle $\DOHL$ with {$(\ol{m}, \eps)$-high-low} partition $(H,L)$ of $G=([n],E)$ with $1\le |E|\le \ol{m}$. 
	The algorithm makes $\poly(\log n,1/\eps)$ queries and has the following properties  when $\alpha^{k-1}\le \deg(u)\le \alpha^{k+1}$: 
	\begin{itemize}
		\item If $\deg(u, L) \leq \alpha^{k-\ell}$, {$\emph{\CheckHLFrac}(u, k, \ell, G)$} outputs ``low'' with probability at least 
		$1 - {\eps^2/n^5}$.\vspace{-0.1cm}
		\item If $\deg(u, L) \geq \alpha^{k-\ell+1}$, {$\emph{\CheckHLFrac}(u, k, \ell, G)$} outputs ``high''
		with probability at least $1 - {\eps^2/n^5}$.\vspace{0.07cm}
	\end{itemize}
\end{lemma} 

\begin{lemma}\label{lem:check-lo-degree}
	There is a randomized algorithm $\emph{\CheckLoDegree}\hspace{0.04cm}(u, d,G)$ which takes as input~a vertex $u\in [n]$, a parameter $0<d \leq \alpha^s$, and
	access to an independent set oracle~and an $(\ol{m},\eps)$-high-low oracle
	of a graph $G=([n],E)$ with $1\le |E|\le \ol{m}$. The algorithm makes 
	$({\alpha^s}/d)\cdot \poly(\log n,1/\eps)$\ignore{\enote{Changed $\sqrt{\ol{m}}$ to $\alpha^s$ since this fits better with the complexity bounds on Lemma 4.4}}
	queries to the two oracles and satisfies the following two properties:
	\begin{itemize}
		\item If $\deg(u, L)\ge (1+\eps) d$, then {$\emph{\CheckLoDegree}(u, d, G)$} outputs ``high'' 
		with probability at least $1 - {\eps^2/n^5}$.\vspace{-0.1cm}
		\item If $\deg(u, L) \leq d$, then {$\emph{\CheckLoDegree}(u, d, G)$} outputs ``low''
		with probability at least $1 - {\eps^2/n^5}$.
	\end{itemize}
\end{lemma}

\begin{proofof}{Lemma \ref{lem:oracle-implementation2} Assuming Lemma 
		\ref{lem:high-low-frac-lemma} and \ref{lem:check-lo-degree}}
	Similar to the proof of Lemma \ref{lem:high-low-simul}, we let $q_{*}$ be a large enough integer so that 
	a string $r\in \{0,1\}^{q_{*}}$ can store randomness needed by calls to
	\begin{enumerate}
		\item  
		${\CheckLoDegree}\hspace{0.04cm}(u, \alpha^{i-1},G)$ 
		for all $u\in [n]$ and $i\in [0:s]$;\vspace{-0.06cm}
		\item
		$\CheckHiDegree\hspace{0.05cm}(u, \alpha^k,G)$ for all $u\in [n]$ and $k\in [s+1:\beta]$; and\vspace{-0.06cm}
		\item 
		${\CheckHLFrac}({u}, k, \ell,G)$\ignore{\enote{changed order here again} }for all $u\in [n]$,
		$k\in [s+1:\beta]$ and $\ell\in [\tau]$\ignore{\enote{removing the call to $\ell = 0$.}}.
	\end{enumerate}
	Then it follows from Lemma \ref{lem:CheckHiDegLemma}, \ref{lem:high-low-frac-lemma} and \ref{lem:check-lo-degree} and a union bound that,
	when $\rr\sim\{0,1\}^{q_{*}}$, all these calls return a correct answer (in the sense that no property as
	stated in Lemma \ref{lem:CheckHiDegLemma}, \ref{lem:high-low-frac-lemma} and \ref{lem:check-lo-degree}  is violated) with probability $1-1/n^3$.
	We will refer to such an $r\in \{0,1\}^{{q_{*}}}$\ignore{\enote{this was $n$ before, a minor typo} }as a \emph{good} string{, and will show that given correct outputs to all calls listed above, $\simDl^*$ and $\simDh^*$ can implement an $(\ol{m}, \eps)$-degree oracle for $G$. For the remainder of the proof, we consider any fixed good string $r \in \{0,1\}^{q_{*}}$.}
	
	Before describing the implementation details of $\simDl^*$ and $\simDh^*$, it is helpful to discuss results of running
	all these algorithms (1), (2) and (3) on a vertex $u$ when $r$ is good.
	We first consider a vertex~$u$ with $\DOHL(u)=0$ and thus, $u\in L$ and we have 
	$\deg(u)\le \alpha^{s+1}$.  
	In this case we consider the results of running  
	${\CheckLoDegree}\hspace{0.04cm}(u, \alpha^{i-1},G)$ for each $i\in [0:s]$,
	and write $a_{i-1}\in\{\text{``low''},\text{``high''}\}$ to denote the result;
	we set $a_{s}=\text{``low''}$ by default.  
	Then there are two cases. If $\deg(u,L)=0$, then all $a_i=\text{``low''}$;
	if $1\le \deg(u,L)\le \alpha^{s+1}$, we have $a_{-1}=\text{``high''}$ and {by Lemma~\ref{lem:check-lo-degree}, as well as the fact $r$ is good,  there is 
		a unique $i\in [0:s]$ such that $a_{i-1}=\text{``high''}$ and $a_{i}=\text{``low''}$,
		where $i$ satisfies $\alpha^{i-1} < \deg(u,L) < \alpha^{i+1}$ (which
		intuitively means that we can place $u$ in $L_i$).}\footnote{{More detailed, we note that $a_{-1} = \text{``high''}$ and $a_{s} = \text{``low''}$, so that some index $i \in [s]$ satisfies $a_{i-1} = \text{``high''}$ and $a_i = \text{``low''}$. In order to see this index is unique, note that, if for $i' \neq i$, $a_{i'-1} = \text{``high''}$ and $a_{i'} = \text{``low''}$, then either $i'-1 > i$, or $i' < i-1$, and $\alpha^{i' - 1} < \deg(u, L) < \alpha^{i' + 1}$; however, this contradicts the fact $\alpha^{i - 1} < \deg(u, L) < \alpha^{i+1}$.}}
	
	Next consider a vertex $u$ with $\DOHL(u)=1$ and thus, $u\in H$ and  $\deg(u)\ge \alpha^s$.
	We first consider  
	${\CheckHiDegree}\hspace{0.04cm}(u,\alpha^{k},G)$ for each $k\in [s+1:\beta-1]$
	and use $b_k {\in \{\text{``low'', ``high'' \}}}$ to denote the result;
	we also set $b_\beta=\text{``low''}$ and $b_{s}=\text{``high''}$ by default.
	By Lemma~\ref{lem:CheckHiDegLemma}, as well as the fact $r$ is good, there is a unique $k\in [s+1:\beta]$ such that $b_{k-1}=\text{``high''}$ and $b_{k}=\text{``low''}$, which implies that $\alpha^{k-1} {<} \deg(u ) {<} \alpha^{k+1}$ (so we can place $u$ in $H_k$).\ignore{\enote{The argument I had in the footnote for uniqueness of the indices $k$ (and $i$ for the low degree case) require inequality to get a contradiction.}}
	Next for this particular $k$, we consider ${\CheckHLFrac}\hspace{0.04cm}({u},k, \ell,G)$ for each $\ell\in [\tau]$ and use $c_\ell$ to denote the result;
	we also set $c_0=\text{``low''}$ by default.
	If  all $c_\ell$'s are ``low,'' then $\deg(u,L) {<} \alpha^{k-\tau+1}$ (which we can
	place in $H_{k,\tau}$).
	Otherwise there exists a {unique} $\ell\in [0:\tau-1]$ such that $c_{\ell}=\text{``low''}$ and $c_{\ell+1}=\text{``high.''}$
	In this case we have $\alpha^{k-\ell-1} {<} \deg(u,L) {<} \alpha^{k-\ell+1}$ (which we can place in $H_{k,\ell}$).
	
	We now describe the implementation of $\simDl^*$ and $\simDh^*$ and show that for every good string $r$,
	they together become an $(\ol{m},\eps)$-degree oracle of the graph:
	\begin{flushleft}\begin{enumerate}
			\item For $\simDl^*\hspace{0.04cm}(u,i,G,r)$, where $i\in [0:s]$, we first check $\DOHL(u)$ and return $0$ if $\DOHL(u)=1$ (meaning that $u\in H$).
			There are two special cases: $i=0$ and $i=1$. If $i=0$, we just run $a_{-1}=\CheckLoDegree \hspace{0.04cm}(u, \alpha^{-1},G)$ and 
			if $a_{-1}=\text{``low''}$ return $1$, and return $0$ otherwise.\\
			If $i=1$, run $a_{-1}=\CheckLoDegree \hspace{0.04cm}(u, \alpha^{-1},G)$ and
			$a_1=\CheckLoDegree \hspace{0.04cm}(u, \alpha,G)$ and
			if $a_{-1}=\text{``high''}$ and $a_1={\text{``low''}}$ return $1$, and return $0$ otherwise.
			For general $i\ge 2$, we run $a_i=\CheckLoDegree \hspace{0.04cm}(u, \alpha^i,G)$
			and $a_{i-1}=\CheckLoDegree \hspace{0.04cm}(u, \alpha^{i-1},G)$
			but set $a_i$ to be $\text{``low''}$ by default if $i=s$.
			If $a_{i-1}=\text{``high''}$ and $a_i=\text{``low''}$, return $1$; otherwise, return $0$.
			\item For $\smash{\simDh^*(u,k,\ell,G,r)}$, where
			$k\in [s+1:\beta]$ and $\ell\in [0:\tau]$,
			we first check $\DOHL(u)$ and return $0$ if $\smash{\DOHL(u)=0}$ (meaning that $u\in L$).
			Next run $b_k= \CheckHiDegree\hspace{0.04cm}(u, \alpha^k,G)$, 
			and $b_{k-1}=\CheckHiDegree\hspace{0.04cm}(u, \alpha^{k-1},G)$
			but set $b_{k-1}=\text{``high''}$ if $k=s+1$ by default and
			set $b_{k}=\text{``low''}$ if $k=\beta$ by default. 
			If $b_{k-1}=\text{``high''}$ and $b_{k}=\text{``low''}$, we continue;
			otherwise we return $0$ (meaning that $u$ does not even belong to $H_k$).
			Finally we run 
			$c_\ell={\CheckHLFrac}\hspace{0.04cm}({u}, k, \ell,G)$ and $c_{\ell+1}={\CheckHLFrac}\hspace{0.04cm}({u}, k, \ell+1,G)$
			but set $c_\ell=\text{``low''}$ by default if $\ell=0$.
			If $\ell=\tau$ and $c_\tau=\text{``low''}$, return $1$; return $0$ otherwise.
			If $\ell<\tau$, return $1$ if
			$c_{\ell}=\text{``low''}$ and $c_{\ell+1}=\text{``high''}$, and return $0$ otherwise.
	\end{enumerate}\end{flushleft}
	Given results of these calls analyzed above, it can be verified that 
	$\simDl^*$ and $\simDh^*$ together implement an $(\ol{m},\eps)$-degree oracle when
	$r$ is a good string. This finishes the proof.
\end{proofof}

{We now provide a proof of Lemma~\ref{lem:oracle-implementation} by using Lemma~\ref{lem:oracle-implementation2} and Lemma~\ref{lem:high-low-simul}.
	
	\begin{proofof}{Lemma~\ref{lem:oracle-implementation}}
		Let $q_{\mathsf{HL}} = q_{\mathsf{HL}}(\eps, n, \ol{m})$ be the integer obtained from Lemma~\ref{lem:high-low-simul}, and $q_{*} = q_{*}(\eps, n, \ol{m})$ be the integer obtained from Lemma~\ref{lem:oracle-implementation2}. We let $q = q_{\mathsf{HL}} + q_{*}$, and we consider a string $\boldr \sim \{0,1\}^q$ defined as the concatenation of $\boldr_1 \sim \{0,1\}^{q_{\mathsf{HL}}}$ and $\boldr_2 \sim \{0,1\}^{q_{*}}$. 
		
		If the function $\HL(\cdot, G, r_1) \colon [n] \to \{0,1\}$ is an $(\ol{m}, \eps)$-high-low oracle of $G$, we say that $r_1$ is a \emph{good} string, and note that by Lemma~\ref{lem:high-low-simul}~ $\boldr_1 \sim \{0,1\}^{q_{\mathsf{HL}}}$ is a good string with probability at least $1 - 1/n^3$. Furthermore, for any fixed $r_1$ which is good, we let $r_2 \in \{0,1\}^{*}$ be  a \emph{good} string if the functions $\simDl^*(\cdot, \cdot, G, r_2) \colon [n] \times [0 : s] \to \{0,1\}$ and $\simDh^*(\cdot, \cdot, \cdot, G, r_2) \colon [n] \times [s+1 : \beta] \times [0:\tau] \to \{0,1\}$, when run with access to the independent set oracle $\IS$ of $G$ and the $(\ol{m}, \eps)$-high-low oracle given by $\HL(\cdot, G, r_1)$, form an $(\ol{m}, \eps)$-degree oracle of $G$. Similarly, by Lemma~\ref{lem:oracle-implementation2},  we have that $\boldr_2 \sim \{0,1\}^{q_{*}}$ is a good string with probability at least $1 - 1/n^{3}$.
		
		As a result, for $\boldr_1$ which is good, and $\boldr_2$ is good (with respect to $\boldr_1$), which occurs with probability $1 - 2/n^3$, the functions $\simDl(\cdot, \cdot, G, \boldr) \colon [n] \times [0 : s] \to \{0,1\}$ and $\simDh(\cdot, \cdot, \cdot, G, \boldr) \colon [n] \times [s+1:\beta] \times [0 : \tau] \to \{0,1\}$ are implemented by calling the functions $\simDl^*$ and $\simDh^*$.
		We note that these functions form an $(\ol{m}, \eps)$-degree oracle of $G$ which makes queries only to the independent set oracle $\IS$ of $G$. 
		
		Lastly, the upper bound on the query complexities to $\IS$ of $\simDl$ and $\simDh$ follows from the upper bounds on the query complexities of $\simDl^*$ and $\simDh^*$ to $\IS$ and $\HL$, as well as the fact that $\HL$ makes at most $\poly(\log n, 1/\eps)$ queries to $\IS$.
\end{proofof}}




\def\LB{\texttt{Low-Bucket}}



\subsection{Proof of Lemma \ref{lem:high-low-frac-lemma}}

We describe {\CheckHLFrac} in Figure \ref{fig:high-to-low-frac}. The procedure shares resemblance with $\CheckHiDegree$
and the main difference is that every time a set $\bT$ is found such that $\bT$ is {an} independent {set} but 
$\bT\cup \{u\}$ is not, we continue to find an edge $(u,v)\in E$ and then use the 
high-low oracle to certify that $v\in L$.
Note that we do not need to run the randomized binary search in order to find an edge $(u,v)\in E$.
Given that $\bT$ is an independent set but $\bT\cup\{u\}$ is not,
one can deterministically split $\bT$ into two parts,
query the two parts together with $u$ separately, and continue with one that is not independent.

\begin{figure}[t!]
	\begin{framed}
		\noindent Subroutine $\CheckHLFrac\hspace{0.05cm}({u}, k,\ell,G)$
		\begin{flushleft}
			\noindent {\bf Input:} {A vertex $u \in [n]$ satisfying $\alpha^{k-1} \leq \deg(u) \leq \alpha^{k+1}$}, integers $k\in [s+1:\beta]$ and $\ell\in [0:\tau]$,
			and access to  an independent set oracle and an $(\ol{m},\eps)$-high-low oracle $\DOHL$ 
			of $G=([n],E)$ with $1\le |E|\le \ol{m}$.\\
			{\bf Output:} Either ``low'' or ``high.'' 
			
			\begin{enumerate}
				\item Let $c$ be a counter, initially set to $0$.
				Repeat for $t =\poly(\log n,1/\eps)$ iterations: 
				\begin{itemize}
					\item Sample $\bT \subseteq [n] \setminus \{ u \}$ by including each element independently with probability $\eps/(\alpha^k\log n)$. If $\bT$ is an independent set but $\bT \cup \{u\}$ in not {(obtained by querying $\IS$)},
					run a deterministic binary search to find an edge $(u, v) \in E$. 
					\item Query $\DOHL(v)$, and increment $c$ if it outputs $0$.\vspace{0.05cm}
				\end{itemize}
				\item If $c > (1+\eps/4)  \eps t \big/ (\alpha^\ell \log n)$, output ``high;'' otherwise, output ``low.'' 
			\end{enumerate}
		\end{flushleft}\vskip -0.14in
	\end{framed}\vspace{-0.2cm}
	\caption{Description of the $\CheckHLFrac$ subroutine.} \label{fig:high-to-low-frac}
\end{figure}

Now we start to prove Lemma \ref{lem:high-low-frac-lemma}.
Consider first the case of $\deg(u, L) \leq \alpha^{k-\ell}$. We note that in any iteration of line 1, the probability $c$ is incremented is at most the probability that a neighbor $v \in \Gamma(u, L)$ is included{, and this} occurs with probability at most
\[ \frac{\eps}{\alpha^k \log n} \cdot \alpha^{k-\ell } \leq \frac{\eps  }{\alpha^{\ell}\log n}. \]

Suppose, on the other hand, that $\deg(u, L)\ge \alpha^{k-\ell+1}$. A sufficient condition for the counter $c$ to be incremented is (1) $\bT$ is an independent set, (2) $\bT$ contains a {unique}\ignore{\enote{this is what we lower-bound the probability by}} neighbor $v \in \Gamma(u, L)$, and (3) $\bT$ avoids all vertices in $\Gamma(u, H)$. Representing $\bT = \bT_1 \cup \bT_2 \cup \bT_3$ where $\bT_1 \subseteq \Gamma(u, L)$, $\bT_2 \subseteq \Gamma(u, H)$, and $\bT_3 \subseteq [n] \setminus \Gamma(u)$, we have:
\begin{align}
	\Prx_{\bT}\left[c \text{ is incremented} \right] &\geq \Prx_{\bT_1, \bT_2, \bT_3}\left[|\bT_1| = 1 \wedge \bT_2 = \emptyset \wedge \bT_1 \cup \bT_3 \text{ is an independent set}  \right] \nonumber \\
	&\geq \Prx_{\bT_2}[\bT_2 = \emptyset]\left(\sum_{v \in \Gamma(u, L)} \Prx_{\bT_1, \bT_3}\left[  \begin{array}{c} \bT_1 = \{ v\} \wedge \\ \bT_1 \cup \bT_3 \text{ is an independent set}\end{array} \right] \right), \label{eq:haha4} \\[0.8ex]
	\Prx_{\bT_2}\left[ \bT_2 = \emptyset\right] &\geq \left(1 - \frac{\eps}{\alpha^k\log n} \right)^{\alpha^{k+1}} \geq 1 - o(\eps). \label{eq:haha5}
\end{align}
We note that since $\deg(u)\le \alpha^{k+1}$, for any $v \in \Gamma(u, L)$,
\begin{align} 
	\Prx_{\bT_1}\left[\bT_1 = \{ v\} \right] &\geq \frac{\eps}{\alpha^k\log n} \left(1 - \frac{\eps}{\alpha^k \log n}  \right)^{\alpha^{k+1}{-1}} \geq \frac{\eps (1-o(\eps))}{\alpha^k \log n}. \label{eq:haha6}
\end{align} 
Finally, conditioning on $\bT_1 = \{ v \}$, $\bT_1 \cup \bT_3$ is an independent set if and only if $\bT_3 \cap \Gamma (v) = \emptyset$ and $\bT_3$ {(which is sampled from $[n] \setminus \Gamma(u)$ and avoids $\Gamma(v)$)} is an independent set. Since $v \in L$, the probability of $\bT_3 \cap \Gamma(v) = \emptyset$ is at least $(1 - \eps / (\alpha^k\log n))^{\alpha^{s+1}}  \geq 1 - o(\eps)$. {As a result, viewing $\bT_3 = \bT_3^{(0)} \cup \bT_3^{(1)}$ where $\bT_3^{(0)} \subset \Gamma(v) \setminus \Gamma(u)$ and $\bT_3^{(1)} \subset [n] \setminus (\Gamma(u) \cup \Gamma(v)$, we have that for any fixed $v \in \Gamma(u, L)$, 
	\begin{align} 
		\Prx_{\bT_1, \bT_3}\left[ \{ v \} \cup \bT_3 \text{ is an independent set} \right] &\geq (1 - o(\eps)) \Prx_{\bT^{(1)}_3}\left[ \bT_3^{(1)} \text{ is an independent set}\right] \geq 1-o(\eps), \label{eq:hahahe}
	\end{align}
	where we used Lemma~\ref{lem:indset} to say $\bT_3^{(1)}$ is an independent set with probability at least $1 - o(\eps)$. Plugging (\ref{eq:haha5}), (\ref{eq:haha6}) and (\ref{eq:hahahe}) back into (\ref{eq:haha4}), and recalling that $|\Gamma(u, L)| = \deg(u, L) \geq \alpha^{k-l+1}$, the probability the counter $c$ is incremented is at least
	\[ \left(1-o(\eps) \right) \cdot \alpha^{k-\ell+1} \cdot \frac{\eps (1-o(\eps))}{\alpha^k \log n} \cdot (1 - o(\eps)) \geq \frac{\eps (1 + \eps / 2)}{\alpha^{\ell}\log n}. \]}\ignore{\enote{just being a bit more formal in this lemma}}
Given that $\alpha^\ell\le \alpha^\tau=O(\log^2 n/\eps)$, it follows from
a Chernoff bound that $\poly(\log n,1/\eps)$ iterations are enough for
the counter to distinguish these two cases with probability at least $1 - \eps^2/n^5$.


\subsection{Proof of Lemma \ref{lem:check-lo-degree}}
\begin{figure}[t!]
	\begin{framed}
		\noindent Subroutine $\CheckLoDegree\hspace{0.05cm}(u, d,G)$
		\begin{flushleft}
			\noindent {\bf Input:} A vertex $u \in L$, a {parameter} $0 \leq d \leq \alpha^{s+1}$, and 
			query access to independent set oracle $\IS$ and an $(\ol{m},\eps)$-degree oracle 
			of a graph $G=([n],E)$ with $1\le |E|\le \ol{m}$.\\
			{\bf Output:} Either ``low'' or ``high.'' 
			
			\begin{enumerate}
				\item Let $c=0$ be a counter. Repeat for $t = ({\alpha^s}/d)\cdot \poly(\log n,1/\eps)$ many iterations:
				\begin{itemize}
					\item Sample $\bT \subseteq [n] \setminus \{ u \}$ by including each vertex independently with probability $\eps/({\alpha^s}\log n)$.
					If $\bT$ is an independent set but $\bT \cup \{ u \}$ is not {(obtained by querying $\IS$)}, run a deterministic binary search to find an edge $(u, v)\in E$ with $v \in \bT$. 
					\item Query $\DOHL(v)$, and increment the counter
					if it outputs $0$, i.e., $v \in L$.
				\end{itemize}
				\item If {$c > (1+\eps/4) \eps d t\big/ ({\alpha^s} \log n)$}, output ``high;'' otherwise, output ``low.''
			\end{enumerate}
		\end{flushleft}\vskip -0.14in
	\end{framed}\vspace{-0.2cm}
	\caption{Description of the $\CheckLoDegree$ subroutine.} \label{fig:low-degree-check}
\end{figure}

We present the algorithm in Figure \ref{fig:low-degree-check}. The proof follows a similar path as 
that of Lemma \ref{lem:high-low-frac-lemma} {with a few parameters set differently}.

{Suppose $\deg(u, L) \leq d$, the probability that the counter is incremented is at most the probability that a neighbor $v \in \Gamma(u, L)$ is included in $\bT$, which occurs with probability at most $\eps d/({\alpha^s}\log n)$. }\ignore{\enote{presenting this first to have parallels with Lemma~\ref{lem:high-low-frac-lemma}.}}
Suppose $\deg(u, L)\ge (1+\eps)d $, and consider the probability, over the draw of $\bT \subseteq [n] \setminus \{u\}$ that the counter $c$ is incremented at any particular round. {Similarly to the proof of Lemma~\ref{lem:high-low-frac-lemma}}, we note that a sufficient condition for this to occur is when (1) $\bT$ is an independent set, (2) $\bT$ contains a {unique} neighbor $v \in \Gamma(u, L)$,\ignore{\enote{this is what we are lower-bounding the probability by}} and (3) $\bT$ avoids all vertices in $\Gamma(u, H)$. Viewing $\bT = \bT_1 \cup \bT_2 \cup \bT_3$ where $\bT_1 \subseteq \Gamma(u, L)$, $\bT_2 \subseteq \Gamma(u, H)$, and $\bT_3 \subseteq [n] \setminus \Gamma(u)$, we have:
\begin{align}
	\Prx_{\bT}\left[c \text{ is incremented} \right] &\geq \Prx_{\bT_1, \bT_2, \bT_3}\left[|\bT_1| = 1 \wedge \bT_2 = \emptyset \wedge \bT_1 \cup \bT_3 \text{ is an independent set}  \right] \nonumber \\
	&\geq \Prx_{\bT_2}[\bT_2 = \emptyset]\left(\sum_{v \in \Gamma(u, L)} \Prx_{\bT_1, \bT_3}\left[  \begin{array}{c} \bT_1 = \{ v\} \wedge \\ \bT_1 \cup \bT_3 \text{ is an independent set}\end{array} \right] \right), \label{eq:haha} \\[0.8ex]
	\Prx_{\bT_2}\left[ \bT_2 = \emptyset\right] &\geq \left(1 - \frac{\eps}{{\alpha^s} \log n} \right)^{\alpha^{s+1}} \geq 1 - o(\eps). \label{eq:haha2}
\end{align}
Next for each $v \in \Gamma(u, L)$, we have
\begin{align} 
	\Prx_{\bT_1}\left[\bT_1 = \{ v\} \right] &\geq \frac{\eps}{{\alpha^s}\log n} \left(1 - \frac{\eps}{
		{\alpha^s} \log n}  \right)^{\alpha^{s+1}{-1}} \geq \frac{\eps (1-o(\eps))}{{\alpha^s}\log n}. \label{eq:haha3}
\end{align} 
Conditioning on $\bT_1 = \{ v \}$, $\bT_1 \cup \bT_3$ is an independent set if and only if $\bT_3 \cap \Gamma (v) = \emptyset$, and $\bT_3$ is an independent set. Similarly to (\ref{eq:haha2}), since $v \in L$, the probability of $\bT_3 \cap \Gamma(v) = \emptyset$ is at least $(1 - \eps / ({\alpha^s}\log n))^{\alpha^{s+1}}  \geq 1 - o(\eps)$. By Lemma~\ref{lem:indset}, $\bT_3$ {(after avoiding $\Gamma(v)$)} is an independent set with probability at least $1 - o(\eps)$. Therefore, we obtain that (\ref{eq:haha}) is at least
\[ \frac{\eps d (1+\eps) (1-o(\eps))^3}{{\alpha^s}\log n} >\frac{\eps d(1+\epsilon/2)}{\alpha^s\log n}\]
from combining (\ref{eq:haha2}) and (\ref{eq:haha3}). 

By a Chernoff bound, the counter will distinguish these two cases with probability at least $1 - \eps^2/n^5$. 

\def\HB{\texttt{High-Bucket}}

\def\HLF{\texttt{High-To-Low-Fraction}}

\ignore{
	In this section we will prove the following main lemma:
	\begin{lemma}\label{lem:check-degree}
		For any graph $G = ([n], E)$ and $\mbar \geq |E|$, as well as $\eps \in (0,1)$, there exists a distribution $\calD$ supported on partitions $(H_{1}, \dots, H_{k}, L_{1}, \dots, L_{\ell})$ of $[n]$ with $k , \ell \leq O(\log n/\eps)$, and two subroutines $\CheckH$ and $\CheckL$ satisfying the following: with probability at least $1 - 1/\poly(n)$ over $(\bH_1,\dots, \bH_{k}, \bL_1, \dots, \bL_{\ell}) \sim \calD$, letting $\bH = \bH_1 \cup \dots \cup \bH_{k}$ and $\bL = \bL_1 \cup \dots \cup \bL_{\ell}$:
		\begin{enumerate}
			\item \label{def:HighBucket} For every $i \in [k]$ and every $u \in \bH_{i}$, $\sqrt{\mbar}(1+\eps)^{i-1} \leq \deg(u) < \sqrt{\mbar}(1+\eps)^{i+1}$.
			\item \label{def:lowBucket} For every $i \in [\ell]$ and every $u \in \bL_{i}$, $(1+\eps)^{i-1} \leq \deg(u, \bL) \leq (1+\eps)^{i+1}$.
			\item \label{lem:CheckH}For each $u \in [n]$ and $i \in [k]$, the subroutine $\CheckH(u, i)$ makes $O(\log^2 n / \eps^3)$ independent set queries and outputs ``yes'' if and only if $u \in \bH_i$. 
			\item \label{lem:checkL} For each $u \in [n]$ and $i \in [\ell]$, the subroutine $\CheckL(u, i)$ makes $O(\sqrt{\ol{m}}\log^4 n / ((1+\eps)^{i} \eps^6))$ independent set queries and outputs ``yes'' if and only if $u \in \bL_i$. 
		\end{enumerate}
	\end{lemma}
	
	\newcommand{\CheckHL}{\texttt{CheckH-L}}
	
	\begin{lemma}\label{lem:check-hi-to-lo}
		Fix a graph $G = ([n], E)$ with $\mbar \geq |E|$, $\eps \in (0,1)$, as well as a draw $(\bH_1, \dots, \bH_{k}, \bL_1, \dots, \bL_{\ell}) \sim \calD$ satisfying the conclusions of Lemma~\ref{lem:check-degree}. For every $i \in [k]$, there exists a distribution $\calD_{HL}(\bH_i)$ supported on partitions $(H_{i, 1}, \dots, H_{i, l})$ of $\bH_i$ with $l \leq O(\log n/ \eps)$ and a subroutine $\CheckHL$ satisfying the following: with probability at least $1 - 1/\poly(n)$ over $(\bH_{i, 1}, \dots, \bH_{i, l}) \sim \calD_{HL}(\bH_i)$,
		\begin{enumerate}
			\item For every $j \in [l]$ and $u \in \bH_{i,j}$, $\sqrt{\mbar} (1+\eps)^{i-j-2} \leq \deg(u, \bL) \leq \sqrt{\mbar} (1+\eps)^{i-j+2}$.
			\item For each $u \in \bH_i$ and $j \in [l]$, the subroutine $\CheckHL(u, i, j)$ makes $O(...)$ independent set queries and outputs ``yes'' if and only if $u \in \bH_{i,j}$.
		\end{enumerate}
	\end{lemma}
	
	\subsection{Proof of Lemma~\ref{lem:check-degree}}
	
	We will now define and prove the properties of the distribution $\calD$ over partitions satisfying the conclusions of Lemma~\ref{lem:check-degree}. In order to define the distribution over partitions $(\bH_1, \dots, \bH_k, \bL_1, \dots, \bL_{\ell})$ of $[n]$, we will first define a distribution over partitions $(\bH, \bL)$ of $[n]$, and then refine the partition into two, $(\bH_1, \dots, \bH_{k})$ of $\bH$ and $(\bL_1, \dots, \bL_{\ell})$ of $\bL$, to give the desired distribution. At a high level, the distribution over the partitions, $(\bH, \bL)$, is described by first giving a randomized algorithm, $\CheckHiDegree$, which outputs, for each vertex $u$  the choice of whether $u \in \bH$ or $u \in \bL$. As a result, one can think of the distribution over partitions $(\bH, \bL)$ as fixing the randomness of executions of $\CheckHiDegree$ for all vertices. We then similarly refine the partition $(\bH_1, \dots, \bH_k)$ and $(\bL_1, \dots, \bL_{\ell})$ by defining other subroutines, $\CheckH$ and $\CheckL$. The first two conclusions in Lemma~\ref{lem:check-degree} are then given by analyzing the properties of the randomized algorithms which define the partition.
	

	\begin{lemma}\label{lem:high-deg}
		There exists a distribution $\calD_0$ supported on partitions $(H, L)$ of $[n]$ such that with probability at least $1 - 1/\poly(n)$ over the draw of $(\bH, \bL) \sim \calD_0$:
		\begin{enumerate}
			\item For every $u \in \bH$, $\deg(u) \geq \sqrt{\mbar}$.
			\item For every $u \in \bL$, $\deg(u) \leq \sqrt{\mbar}(1+\eps)$.
			\item For every vertex $u$, the \emph{deterministic} algorithm $\CheckHiDegree(u, \sqrt{\mbar}, \eps)$ given by a fixed setting of randomness makes $O(\log^2 n / \eps^3)$ independent set queries and outputs ``high'' if and only if $u \in \bH$ and ``low'' if and only if $u \in \bL$.
		\end{enumerate} 
	\end{lemma}
	
	\begin{proof}
		The distribution $\calD_0$ is given by the following procedure: 1) generate a random function $\rho \colon [n] \to \{0,1\}^*$ where $\brho(u)$ stores the randomness to the call $\CheckHiDegree(u, \sqrt{\mbar}, \eps)$, 2) we let $u \in \bH$ if and only if $\CheckHiDegree(u, \sqrt{\mbar}, \eps)$ with randomness fixed to $\brho(u)$ outputs ``high''. We note that (1) and (2) follow from Lemma~\ref{lem:CheckHiDegLemma}.
	\end{proof}
	
	\begin{lemma}\label{lem:sub-divide-high}
		For any draw $(\bH, \bL) \sim \calD_0$, there exists a distribution $\calD_{\bH}$ supported on partitions $(H_1, \dots, H_k)$ of $\bH$ with $k \leq O(\log (n/\sqrt{\mbar})/\eps)$ such that with probability at least $1 - 1/\poly(n)$ over the draw of $(\bH_1, \dots, \bH_k) \sim \calD_{\bH}$:
		\begin{itemize}
			\item For every $i \in [k]$ and every $u \in \bH_i$, $\sqrt{\mbar}(1+\eps)^{i-1} \leq \deg(u) \leq \sqrt{\mbar}(1+\eps)^{i+1}$.
			\item For every vertex $u \in \bH$, there exists a deterministic algorithm $\CheckH(u, i)$ given by a fixed setting of the randomness making $O(\log^2 n/\eps^3)$ independent set queries which outputs ``yes'' if and only if $u \in \bH_i$.
		\end{itemize}
	\end{lemma}
	
	\begin{proof}
		For $(\bH, \bL) \sim \calD_0$ satisfying the conclusions of Lemma~\ref{lem:high-deg}, a draw $(\bH_1, \dots, \bH_k) \sim \calD_{\bH}$ is given by sampling a random function $\brho \colon \bH \times [k] \to \{0,1\}^*$ which holds the randomness $\brho(u, i)$ for the call $\CheckH(u, i)$, and letting $u \in \bH_i$ if $i$ is the unique index where $\CheckH(u, i)$ outputs ``yes''. Via a union bound over $O(n \log n)$ many executions of $\CheckHiDegree$, with probability at least $1 - 1/\poly(n)$, every output of $\CheckHiDegree(u, d_i, \eps)$ satisfies the conclusions of Lemma~\ref{lem:CheckHiDegLemma}. 
		
		Suppose $\CheckH(u, i)$ outputs ``yes''. By Lemma~\ref{lem:CheckHiDegLemma}, we have that $\sqrt{\mbar}(1+\eps)^{i-1} \leq \deg(u) \leq \sqrt{\mbar}(1+\eps)^{i+1}$. Finally, in order to see that $(\bH_1, \dots, \bH_k)$ forms a partition of $\bH$ note that for every $u \in \bH$, $\deg(u) \leq n$, so that for $d_k = \sqrt{\mbar}(1+\eps)^k = n$, $\CheckHiDegree(u, d_k, \eps)$ outputs ``low''; furthermore, $\deg(u) \geq \sqrt{\mbar}$ so that for $d_0 = \sqrt{\mbar} / (1+\eps)$, $\CheckHiDegree(u, d_0, \eps)$ outputs ``high''. Thus, there exists some $i \in [k]$ which outputs ``yes''. Suppose such an index $i$ is not unique, and there exists $j \neq i$ where $\CheckH(u, j)$ outputs ``yes'', then since randomness of all executions of $\CheckHiDegree$ was fixed, we must have either $j > i+1$ or $j < i-1$, and this contradicts the fact that $\sqrt{\mbar}(1+\eps)^{j-1} \leq \deg(u) \leq \sqrt{\mbar}(1+\eps)^{j+1}$.
	\end{proof}
	
	It now remains to refine the random partition from Lemma~\ref{lem:high-deg} in order to prove Lemma~\ref{lem:check-degree}. We now define the sub-routine $\CheckLoDegree$, which similarly to $\CheckHiDegree$, checks whether a particular vertex $u \in \bL$ has a neighborhood \emph{in $\bL$} of size approximately $d$.
	
	Finally, we may combine Lemma~\ref{lem:high-deg}, Lemma~\ref{lem:sub-divide-high}, and Lemma~\ref{lem:lo-degree-subdivide} in order to prove Lemma~\ref{lem:check-degree}. Specifically, a draw $(\bH_1, \dots, \bH_{k}, \bL_1, \dots, \bL_k) \sim \calD$ is given by first sampling $(\bH, \bL) \sim \calD_0$, and then sampling $(\bH_1, \dots, \bH_k) \sim \calD_{\bH}$ and $(\bL_1, \dots, \bL_{\ell}) \sim \calD_{\bL}$. The sub-routine $\CheckHiDegree(\cdot, \sqrt{\mbar}, \eps)$ first assigns a vertex $u$ into $\bH$ or $\bL$ and then the subroutines $\CheckH$ and $\CheckL$ assign vertices $u$ from $\bH$ or $\bL$ into $(\bH_1, \dots, \bH_k)$ or $(\bL_1, \dots, \bL_{\ell})$, respectively.
	
	\subsection{Proof of Lemma~\ref{lem:check-hi-to-lo}}
	
	We proceed similarly to the proof of Lemma~\ref{lem:check-degree} by defining a randomized algorithm which outputs for every $u \in \bH_i$ and every $j \in [l]$, a decision as to whether $u \in \bH_{i, j}$. The distribution is then taken by sampling the randomness needed to run the algorithm, and the guarantees on $u \in \bH_{i,j}$ are given by analyzing the properties of the randomized algorithm.
}

\section{Lower Bound}

We now turn to proving the lower bound on the query complexity of estimating the number of edges of an
undirected graph $G$ with access to the independent set oracle $\IS$.

We restate the main lower bound theorem:  

\begin{theorem}\label{thm:lb}
	Let $n$ and $m$ be two positive integers with $m\le {n\choose 2}$.
	Any randomized algorithm with access to the independent set oracle $\IS$ of an unknown 
	$G=([n],E)$ must make $\min(\sqrt{m},n /\sqrt{m})\cdot (\poly\log n)^{-1}$ many queries in order to distinguish whether $|E|\le m/2$ or~$|E|\ge m$~with probability at least $2/3$.
\end{theorem}

We first establish Theorem \ref{thm:lb} for the case when $m\ge n$;
the case when $m< n$  follows later with a simple reduction to the case when $m\ge n$. 
Now let $m$ be an integer with 
\begin{equation}\label{assumption}
n \le m\le \frac{n^2}{\log^6 n}.
\end{equation}
Note that we further assumed that $m\le n^2/\log^6 n$.
When $\smash{m\ge n^2/\log^6 n}$,  the lower bound we~aim for becomes $\smash{\widetilde{\Omega}(\log^3 n)}$
which holds trivially since (1) $\widetilde{\Omega}$ hides a factor of $\polylog(n)$ 
and (2)
solving the problem requires at least one query to $\IS$ given that $ m\le {n\choose 2}$.

Assuming that $m$ satisfies (\ref{assumption}),
the proof proceeds by Yao's principle.
In Section \ref{sec:dist} we present two distributions $\Dyes$ and $\Dno$ over undirected 
graphs with vertex set $[n]$ such that 
$\bG\sim \Dyes$~has fewer than $m/2$ edges with probability at least $1-o(1)$
and $\bG\sim \Dno$ has more than $m$ edges with probability at least $1-o(1)$.
Next, we prove in Section \ref{sec:proof} that every deterministic algorithm that distinguishes $\Dyes$ and $\Dno$
must make $\smash{\widetilde{\Omega}(n/\sqrt{m})}$ independent set queries.
This finishes the proof of Theorem \ref{thm:lb} when $m\ge n $.
We work on the case when $m< n $ via a reduction in Section \ref{finalsec}.

\subsection{Distributions}\label{sec:dist} 
Let $d\eqdef m/n$ (which is not necessarily an integer). Given that $m$ satisfies (\ref{assumption}), we have that
\begin{equation}\label{assumptiond}
1 \le d\le \frac{n}{\log^6 n}.
\end{equation}
Let $q$ be the following positive integer:
$$
q=\left\lceil \sqrt{\frac{n}{d}}\cdot \frac{1}{\log^3 n}\right\rceil
= \Theta\left(\sqrt{\frac{n}{d}}\cdot \frac{1}{\log^3 n}\right).
$$
We consider the following two distributions supported on graphs with vertex set $[n]$:
\begin{flushleft}\begin{itemize}
		\item $\Dno$: A graph $\bG \sim \Dno$ is sampled by first letting $\bA \subseteq [n]$ be a uniformly random subset
		\\ of $[n]$, and $\ol{\bA} = [n] \setminus \bA$. Furthermore, we sample $\bB \subseteq \bA$ by including each element of $\bA$ in $\bB$ independently with probability $d\log n/n$ (note that this is smaller than $1$ by (\ref{assumptiond})). For each $i \in \bA \setminus \bB$ and $j \in \ol{\bA}$, we include the edge $(i,j)$ in $\bG$ independently with probability $d/n$. Finally, we add the edge $(i,j)$ to $\bG$ for every $i \in \bB$ and $j \in \ol{\bA}$.
		
		\item $\Dyes$: A graph $\bG \sim \Dyes$ is sampled by first letting $\bA \subseteq [n]$ be a uniformly random subset of $[n]$, and $\ol{\bA} = [n] \setminus \bA$  as above.
		We set $\bB=\emptyset$ by default in $\Dyes$.\footnote{We introduce $\bB$ in $\Dyes$ only for the 
			purpose of analysis later.} For each $i \in \bA\setminus \bB=\bA$ and each $j \in \ol{\bA}$, we include the edge $(i,j)$ in $\bG$ independently with probability $d/n$.
\end{itemize}\end{flushleft}

We note that with probability at least $1 - o(1)$ over the draw of $\bG \sim \Dyes$, $\bG$ will have no more than 
$m/2$ many edges. 
This follows from Chernoff bound and the fact that there are at most $n^2/4$ many pairs between $\bA$
and $\ol{\bA}$ (so the expected number of edges is no more than $(n^2/4)\cdot (d/n)=m/4$).
On the other hand, with probability at least $1-o(1)$ over the draw of $\bG \sim \Dno$, $\bG$ will have $\Omega(d n \log n)\ge m$  edges. This is because with probability $1- o(1)$, $|\ol{\bA}| = \Omega(n)$ and $|\bB| = \Omega(d \log n)$. 

As a result, Theorem \ref{thm:lb} (when $m\ge n $) follows from Lemma \ref{mainmainlm} below
because any randomized algorithm that can distinguish $|E|\le m/2$ and $|E|\ge m$ with probability $2/3$
implies a deterministic algorithm $\ALG$ with the same complexity such that
\begin{align*}
	\Prx_{\bG \sim \Dno}\big[\hspace{0.04cm} {\ALG}(\bG) \text{ outputs ``no''}\hspace{0.03cm} \big] - \Prx_{\bG \sim \Dyes}\big[\hspace{0.04cm}  {\ALG}(\bG) \text{ outputs ``no''}\hspace{0.03cm}  \big] \ge 1/3-o(1). 
\end{align*}


\begin{lemma}\label{mainmainlm}
	Let $\emph{\ALG}$ be a deterministic algorithm that makes $q$ independent set queries. Then
	\begin{align*}
		\Prx_{\bG \sim \Dno}\big[\hspace{0.04cm} \emph{\ALG}(\bG) \text{ outputs ``no''}\hspace{0.03cm} \big] - \Prx_{\bG \sim \Dyes}\big[\hspace{0.04cm}  \emph{\ALG}(\bG) \text{ outputs ``no''}\hspace{0.03cm}  \big] \leq o(1). 
	\end{align*}
\end{lemma}

\subsection{Augmented oracle}\label{sec:proof}

To prove Lemma \ref{mainmainlm},  we will work with an \emph{augmented} (independent set) oracle.
We show that any deterministic algorithm with access to the original independent set oracle
can be simulated exactly using the augmented oracle with the same query complexity (Lemma \ref{simulationlem}).
As a result lower bounds for the augmented oracle (Lemma \ref{lem:augmented})
carry over to the independent set oracle (Lemma \ref{mainmainlm}). 

The augmented oracle is specifically designed to be queried when the input graph
is drawn from either $\Dyes$ or $\Dno$. 
Suppose that $\bG$ is drawn from $\Dyes$ or $\Dno$ together with
the auxiliary sets $\bA$ and $\bB$ (see Section \ref{sec:dist}).
A deterministic algorithm can access the augmented oracle as follows:
\begin{flushleft}\begin{itemize}
		\item 
		At any time during its execution,
		the algorithm maintains a triple $(K,\ell,e)$ which we will refer to as its current \emph{knowledge triple},
		where $K\subseteq [n]$ is a set of vertices, $\ell \colon K \to \{ \ol{a},a, b\}$
		assigns one of three labels to each vertex in $K$,
		and $e \colon K \times K \to \{0,1\}$.
		We refer to vertices in $K$ as \emph{known} vertices.
		Initially, $K = \emptyset$ (and both $\ell$ and $e$ are trivial) and will grow as the result of queries made by the algorithm to the
		augmented oracle (see the next paragraph).
		For each vertex $i \in K$, $\ell(i)$ indicates whether $i \in \ol{\bA}, \bA \setminus \bB$ or $\bB$: if $i \in \ol{\bA}$, then $\ell(i) = \ol{a}$;  if $i \in \bA \setminus \bB$, then $\ell(i) = a$;
		if $i \in \bB$, then $\ell(i) = b$.\footnote{Recall that when $\bG \sim \Dyes$, we set $\bB = \emptyset$ by default. As a result, $\ell(i)=b$ can never happen when $\bG\sim\Dyes$.} Moreover, for any vertices
		$i, j \in K$, $e(i,j)$ is the indicator of whether $(i,j)$ lies in $\bG$ or not.
		\item 
		At the beginning of each round, based on its current knowledge triple $(K,\ell,e)$, 
		the algorithm can deterministically send a query specified by 
		a set $Q \subseteq [n] \setminus K$ to the augmented oracle.
		The oracle then reacts to the query as follows:
		\begin{itemize}
			\item If $|Q| \leq t$, where $t$ denotes  the following integer
			parameter$$t\eqdef \left\lceil \sqrt{n/d}\cdot \log n\right\rceil=\Theta\left(\sqrt{n/d}\cdot \log n\right),$$ the oracle sends a new knowledge triple
			$(K,\ell,e)$ to the algorithm with $K \leftarrow K \cup Q$ and with both $\ell$ and $e$ updated 
			according to $\bG$, $\bA$ and $\bB$.\vspace{0.1cm} 
			\item If $|Q| > t$, the oracle samples a subset $\bL \subseteq Q$ of size $t$ uniformly at random.
			If $\bL$ is not an independent set of $\bG$, the oracle sends 
			a new knowledge triple $(K,\ell,e)$ to the algorithm with $K\leftarrow K\cup \bL$.
			If $\bL$ happens to be an independent set of $\bG$, we say the oracle ``\emph{fails}''   and it 
			sends a new knowledge triple $(K,\ell,e)$ with $K\leftarrow [n]$ (i.e., in this case the
			oracle simply gives up and sends the whole graph to the algorithm). \end{itemize}
\end{itemize}\end{flushleft}
Note that even when the algorithm is deterministic,
the augmented oracle is randomized due to $\bL$.

We show that any algorithm with access to the original independent set oracle 
can be simulated using the augmented oracle with the same query complexity.
\begin{lemma}\label{simulationlem}
	Let $\emph{\ALG}$ be a deterministic algorithm with access to the independent set oracle.
	Then there is a deterministic algorithm $\emph{\ALG}^*$ with access to the augmented oracle (running
	over 
	$\bG$ drawn from either $\Dyes$ or $\Dno$ only) such that
	$\emph{\ALG}^*$ has the same query complexity as $\emph{\ALG}$, 
	$$
	\Prx_{\bG \sim \Dno}\big[\hspace{0.04cm} \emph{\ALG}(\bG) \text{ outputs ``no''}\hspace{0.04cm} \big]
	=\Prx_{\bG \sim \Dno}\big[\hspace{0.04cm} \emph{\ALG}^*(\bG) \text{ outputs ``no''}\hspace{0.04cm} \big],\footnote{Note that the second probability is over not only the draw of $\bG\sim \Dno$ but also
		the randomness of the augmented oracle. The same comment applies to similar expressions in the rest of the section.}
	$$
	and the same equation holds for $\Dyes$.
\end{lemma}
\begin{proof}
	The algorithm $\ALG^*$ simulates $\ALG$ query by query as follows.
	Let $(K,\ell,e)$ be the current knowledge triple of $\ALG^*$ (with $K=\emptyset$ initially),
	and let $S\subseteq [n]$ be the next query of $\ALG$.
	So $\ALG^*$ needs~to know if $S$ is an independent set or not in order to continue 
	the simulation of $\ALG$.
	
	For this purpose $\ALG^*$ queries $S \setminus K$ to the augmented oracle.
	If $|S \setminus K| \leq t$, 
	$\ALG^*$ will receive~an updated knowledge triple from the augmented oracle with $K\leftarrow S\cup K$.
	With the updated $e$, $\ALG^*$ can determine if $S$ is an independent set or not and
	continue the simulation of $\ALG$.
	
	On the other hand,  when $|S \setminus K|> t$, one of the following two events will occur:
	either  1)  $\ALG^*$ will receive an update from the augmented oracle
	with $K \leftarrow K \cup \bL$, where $\bL \subseteq (S \setminus K)$ is not an independent set,
	(2) or the oracle ``fails'' and $\ALG^*$ receives the whole graph.
	In the first case, $\ALG^*$ knows that $S$ is not an independent set and can continue the
	simulation.
	In the second case, $\ALG^*$ can use the graph to finish the simulation of $\ALG$.
	The above simulation of $\ALG$ uses no more queries that $\ALG$ itself.
	This finishes the proof of the lemma.
\end{proof}\medskip

Given Lemma \ref{simulationlem}, Lemma \ref{mainmainlm} follows directly from the following lemma:

\begin{lemma}\label{lem:augmented}
	Let $\emph{\ALG}^{*}$ be any  deterministic algorithm that makes $q$ queries to the augmented oracle
	(over graphs $\bG$ drawn from either $\Dyes$ or $\Dno$ only). Then we have
	\[ \Prx_{\substack{\bG \sim \Dno}}\big[\hspace{0.04cm} \emph{\ALG}^*(\bG) \text{ outputs ``no''}\hspace{0.03cm} \big] - \Prx_{\substack{\bG \sim \Dyes }}\big[\hspace{0.04cm} \emph{\ALG}^*(\bG) \text{ outputs ``no''}\hspace{0.03cm} \big] \leq o(1). \]
\end{lemma}

We start with some intuition.  
Given access to the augmented oracle, 
an algorithm will aim to make the set $K$ as large as possible in order to maximize the chance of 
$\bB \cap K \neq \emptyset$ (in~which~case one can conclude that $\bG \sim \Dno$). 
Now if the algorithm makes a query $Q\subseteq [n]\setminus K$ with
$|Q|\le t=$ $\sqrt{n/d}\cdot \log n$, the probability of $Q\cap \bB\ne \emptyset$ is $o(1/q)$ given that 
$|\bB|$ is only roughly $d\log n$ (this will be made formal in the proof of Lemma \ref{lem:bad}).
On the other hand, if $|Q|>t$, then a vertex in $\bB$
can be added to $K$ because either $\bL\cap \bB\ne \emptyset$ (which happens 
with low probability by a similar analysis since $|\bL|=t$) or the oracle ``fails'' (which 
we show that is unlikely to happen).


To proceed with the proof of Lemma \ref{lem:augmented} we view $\ALG^*$ as a tree of depth $q$ in which each 
internal node is labelled by a query set, each leaf is labelled either ``yes'' or ``no,'' 
and each edge is labelled by
a knowledge triple $(K,\ell,e)$ as the result of the previous query received from the augmented~oracle.
Let $(u,v)$ be an edge with $v$ being a child of $u$.
The label of $(u,v)$ is the current knowledge triple $(K,\ell,e)$
of the algorithm when it arrives at $v$ and thus, the query set $Q$ at $v$ is a subset of $[n]\setminus K$.
We will refer to the label of $(u,v)$ as the current knowledge triple of $v$;
for the root we have $K=\emptyset$.

We introduce the following definition of good and bad nodes, which is
inspired by the intuition that an algorithm would aim for
reaching a $K$ with $K\cap \bB\ne \emptyset$.
We then prove two lemmas based on this definition and use them to prove Lemma \ref{lem:augmented}.

\begin{definition}
	We say a node $v$ in the tree of an algorithm $\emph{\ALG}^*$ is \emph{good} if its current 
	knowledge triple $(K, \ell, e)$ satisfies $\ell^{-1}(b) = \emptyset$.
	We say $v$ is \emph{bad} otherwise.
\end{definition}

We note that $\bG \sim \Dyes$ can never reach a bad node
since $\bB=\emptyset$ in this case. 
\begin{lemma}\label{lem:bad}
	We have (the probability is over $\bG\sim\Dno$ and randomness of
	the augmented oracle) 
	\[ \Prx_{\substack{\bG \sim \Dno}}\big[\hspace{0.04cm}  \emph{\ALG}^*(\bG) \text{ reaches a bad node}\hspace{0.06cm}\big] = o(1). \]
\end{lemma}
\begin{lemma}\label{lem:good}
	For every good node $v$ in the tree of $\emph{\ALG}^*$, we have
	\[ \Prx_{\substack{\bG \sim \Dno }}\big[\hspace{0.04cm} \emph{\ALG}^*(\bG) \text{ reaches $v$}\hspace{0.04cm}\big] \leq \Prx_{\substack{\bG \sim \Dyes }}\big[\hspace{0.04cm} \emph{\ALG}^*(\bG) \text{ reaches $v$}\hspace{0.04cm}\big]. \]
\end{lemma}
We delay the proofs of these two lemmas and first use them to prove Lemma \ref{lem:augmented}.\medskip

\begin{proofof}{Lemma~\ref{lem:augmented} Assuming Lemma \ref{lem:bad} and Lemma \ref{lem:good}}
	Let $N$ denote the set of leaves of $\ALG^*$ that are labelled ``no,''
	and let $N_g\subseteq N$ denote those that are also good.
	Then we have
	\begin{align}
		&\Prx_{\bG \sim \Dno}\big[\hspace{0.04cm} \ALG^*(\bG) \text{ outputs ``no''} 
		\hspace{0.04cm}\big]\\[0.5ex] &\le \sum_{v \in N_g} \Prx_{\bG \sim \Dno}\big[\hspace{0.04cm} \ALG^*(\bG) \text{ reaches $v$}\hspace{0.04cm}\big] + \Prx_{\bG \sim \Dno}\big[\hspace{0.04cm} \ALG^*(\bG) \text{ reaches a bad node}\hspace{0.04cm}\big] \nonumber\\
		&\leq \sum_{v \in N_g} \Prx_{\bG\sim \Dyes}\big[\hspace{0.04cm} \ALG^*(\bG) \text{ reaches $v$}\hspace{0.04cm}\big] + o(1) \label{eq:hehe}\\[0.3ex]
		&\leq \Prx_{\bG \sim \Dyes}\big[\hspace{0.04cm} \ALG^*(\bG) \text{ outputs ``no''}\hspace{0.04cm}\big] + o(1), \nonumber
	\end{align}
	where we used both Lemma~\ref{lem:bad} and Lemma~\ref{lem:good} in (\ref{eq:hehe}).
\end{proofof}\medskip

Finally we prove Lemma \ref{lem:bad} and Lemma \ref{lem:good}.
We start with a lemma that, given
any  node~$v$~in the tree of $\ALG^*$,  describes exactly the
distribution~of $\bG \sim \Dyes$ (or $\bG\sim\Dno$) conditioning on 
$\ALG^*$ arriving at $v$.
Roughly speaking,  
this conditional distribution 
is given by including all edges and labelings of $K$ indicated in $\ell$ and $e$,
and otherwise assigning vertices and edges independently as in the constructions of $\Dyes$ and $\Dno$. 

\begin{lemma}\label{lem:distribution}
	Let $v$ be a node of $\emph{\ALG}^*$ and let $(K,\ell,e)$ be its current knowledge triple.
	Then \begin{flushleft}\begin{itemize}
			\item When $v$ is a good node (recall that $\bG\sim\Dyes$ can never 
			reach a bad node), a graph $\bG \sim \Dyes$ conditioning on reaching node $v$ can be generated as follows. 
			We assign each $i\in K$ to $\bA$ or $\ol{\bA}$ according to $\ell(i)$.
			For each $i\notin K$, we include $i \in \bA$ independently 
			with probability $1/2$, and otherwise include it in $\ol{\bA}$.
			For any two vertices $i,j \in K$, we include $(i,j)$ as an edge in $\bG$ if and only if $e(i,j) = 1$. 
			{For each pair $(i,j)$ such that $i\in \bA$, $j\in \ol{\bA}$ and at least one of $i,j$ does not lie in $K$,
				we include $(i,j)$ as an edge in $\bG$ independently with probability $d/n$.}
			\item A graph $\bG \sim \Dno$ conditioning on reaching $v$ can be generated as follows.
			We first assign each vertex $i\in K$ to $\ol{\bA}$, $\bA\setminus \bB$ and $\bB$ 
			according to $\ell(i)$.
			We then include each vertex $i\notin K$ in $\bA$ independently with probability $1/2$,
			and otherwise include it in $\ol{\bA}$.
			For each $i\notin K$ included in $\bA$, we also include it in $\bB$ independently with probability $d\log n/n$.  
			Next for each $i,j\in K$, we include $(i,j)$ as an edge in $\bG$ if and only if $e(i,j)=1$.
			{Finally, for each pair $(i,j)$ such that $i\in \bA$, $j\in \ol{\bA}$ and at least one of $i,j$ does not lie in $K$,
				we include it as an edge in $\bG$ independently with probability $d/n$ if $i\notin \bB$
				and always include it in $\bG$ if $i\in \bB$.}
	\end{itemize}\end{flushleft}
\end{lemma}

\begin{proof}
	We consider the case of $\bG \sim \Dyes$ since the argument will be analogous for the other case. To this end, we may consider generating  $\bG \sim \Dyes$ (the original $\Dyes$ distribution) according to the following process
	with four steps: (1) We first assign each vertex $i \in K$ independently to $\bA$ or $\ol{\bA}$ with probability $1/2$; (2) Then, for each $i \in K \cap \bA$ and $j \in K \cap \ol{\bA}$, we include $(i, j)$ as an edge in $\bG$ independently with probability $d/n$; (3) Next we assign for each $i \notin K$ to $\bA$ or $\ol{\bA}$ with probability $1/2$; (4)~{Finally,~for each $(i,j)$ such that $i\in \bA$, $j\in \ol{\bA}$ 
		and at least one of $i,j$ does not lie in $\bA$,
		we include it as an edge in $\bG$~independently with probability $d/n$.}
	We note that the event of $\bG \sim \Dyes$ reaching~$v$ fixes the randomness of steps (1) and (2). 
	As a result, the conditional distribution of $\Dyes$ described 
	in the lemma matches  the randomness that remains in steps (3) and (4).
\end{proof}\medskip

We are now ready to prove Lemma \ref{lem:bad} and Lemma \ref{lem:good}.\medskip

\begin{proofof}{Lemma \ref{lem:bad}}
	Consider a good node $v$ of $\ALG^*$ and let
	$(K, \ell, e)$ be its current knowledge triple with $\ell^{-1}(b) = \emptyset$. 
	Let $Q \subseteq [n] \setminus K$ be the query made at $v$. 
	We show below that the probability of $\ALG^*$ reaching a child of $v$ that is bad after this query is $o(1/q)$.
	The lemma then follows from a union bound over the $q$ queries of $\ALG^*$.
	
	The analysis considers two cases. If $|Q| \leq t$, then by Lemma~\ref{lem:distribution},
	\begin{align} 
		\Prx_{\bG \sim \Dno}\big[\hspace{0.04cm} Q \cap \bB \neq \emptyset \mid {\ALG}^*(\bG) \text{ reaches $v$}\hspace{0.04cm}\big] \leq |Q| \cdot \frac{1}{2}\cdot \frac{d\log n}{n} \leq \sqrt{d/n}\cdot  \log^2 n &= o\left(1/q\right).  \label{eq:small-bad}
	\end{align}
	If $|Q| >t$, $\ALG^*$ reaches a bad node if $\bL\cap \bB\ne \emptyset$ or 
	$\bL$ is independent. So the probability is at most  
	\begin{align*}
		\Prx_{\substack{\bG \sim \Dno \\ \bL \subseteq Q}}\big[\hspace{0.04cm} \bL \cap \bB \neq \emptyset \mid {\ALG}^*(\bG) \text{ reaches $v$} \hspace{0.04cm}\big] + \Prx_{\substack{\bG \sim \Dno \\ \bL \subseteq Q}}\big[\hspace{0.04cm} \bL \text{ contains no edges} \mid {\ALG}^*(\bG) \text{ reaches $v$} \hspace{0.04cm}\big]. 
	\end{align*}
	We note that the first summand is at most $o(1/q)$, similarly to (\ref{eq:small-bad}). 
	
	For the second summand  we note that with probability at least $1 -{1}/{\poly(n)}$ over the draw~of $\bG \sim \Dno$
	conditioning on reaching $v$ and $\bL \subseteq Q$,  we have both $$|\bL \cap \bA| \geq \Omega\left(\sqrt{n/d}\cdot \log n\right)\quad\text{and}\quad
	|\bL \cap \ol{\bA}| = \Omega\left(\sqrt{n/d}\cdot \log n\right).$$ When this occurs, by Lemma~\ref{lem:distribution}, each $i \in \bL \cap \bA$ and $j \in \bL \cap \ol{\bA}$ satisfies that either $(i,j)$ is always included as an edge
	(as $i\in \bB$) or $(i,j)$ is included as an edge independently with probability $d/n$.
	As a result, the probability there are no edges in $\bL$ is at most $1/\poly(n)=o(1/q)$.
\end{proofof}\medskip

\begin{proofof}{Lemma \ref{lem:good}}
	We consider the following coupling $(\bG,\bG')$ of the two distributions $\Dyes$ and $\Dno$.
	We first draw $\bG\sim \Dyes$ together with the vertex set $\bA$.
	Then we draw a string $\boldr\in \{0,1\}^\bA$ where each bit $\boldr_i$ is set to 
	$1$ with probability $d\log n/n$ independently.
	Finally we convert $\bG$ into $\bG'$ by adding $(i,j)$ as an edge in $\bG'$
	for all $i\in \bA$ with $\boldr_i=1$ and all $j\in \ol{\bA}$.
	It is then clear that the marginal distribution of $\bG'$ is exactly $\Dno$ so this is a coupling of $\Dyes$ and $\Dno$.
	
	Now let $v$ be a good node in the tree of $\ALG^*$ with $(K,\ell,e)$ being its current knowledge triple.
	Then under the coupling $(\bG,\bG')$ described above we have  
	\[ \Prx_{(\bG,\bG')}\big[\hspace{0.04cm} \ALG^*(\bG') \text{ reaches $v$}\hspace{0.04cm}\big]= \Prx_{ (\bG,\bG')}\big[\hspace{0.04cm} \ALG^*(\bG ) \text{ reaches $v$} \wedge \forall\hspace{0.04cm}i \in \ell^{-1}(a), \boldr_i = 0 \hspace{0.04cm}\big], \]
	because with the same randomness (in sampling $\bL$ each time when needed) in the augmented oracle 
	$\ALG^*(\bG')$ would reach the same node $v$ if and only if $\ALG^*(\bG)$ reaches $v$ and 
	$\boldr_i=0$ for all $i\in \ell^{-1}(a)$.
	As a result, the probability of $\bG \sim\Dyes$ reaching $v$ is at least as large as 
	that of $\bG'\sim\Dno$.
\end{proofof}\medskip


\subsection{Proof of Theorem~\ref{thm:lb}}\label{finalsec}

The case when $m\ge n$ follows directly from Lemma \ref{lem:augmented}.

When $m<n$, we use the observation that every randomized algorithm with parameters $n$ and $m$
(i.e., determining whether an input graph $G=([n],E)$ satisfies $|E|\le m/2$ or $|E|\ge m$)~implies a randomized algorithm with parameters $m$ and $m$
(i.e., determining whether a given graph $G'=([m],E')$ has $|E'|\le m/2$ or $|E'|\ge m$)
with the same query complexity by simply embedding the input graph $G'=([m],E')$ in a graph $G=([n],E)$ using its first $m$ vertices 
(and noting that the independent set oracle of $G$ can be simulated using that of $G'$ query by query).
The latter task, by Lemma \ref{lem:augmented}, has a lower bound of $\tilde{\Omega}(\sqrt{m})$.
This finish the proof of the theorem when $m<n$.

\begin{flushleft}
\bibliographystyle{alpha}
\bibliography{waingarten}
\end{flushleft}
\end{document}